\documentclass{article}
\usepackage{arxiv}
\usepackage[utf8]{inputenc}
\usepackage[T1]{fontenc}
\usepackage{graphicx}
\usepackage{amsmath}
\usepackage{amssymb}
\usepackage{amsthm}
\usepackage{amsfonts}
\usepackage{subfig}
\usepackage{appendix}
\usepackage{booktabs}
\usepackage{hyperref}
\usepackage[numbers]{natbib}
\usepackage{placeins}
\usepackage{wrapfig}
\usepackage{float}

\newtheorem{theorem}{Theorem}
\newtheorem{lemma}{Lemma}
\newtheorem{corollary}{Corollary}
\newtheorem*{observation}{Observation}

\title{Difference Rewards Policy Gradients}
\author{Jacopo Castellini\\
Dept. of Computer Science\\
University of Liverpool\\
\texttt{J.Castellini@liverpool.ac.uk}\\
\And
Sam Devlin\\
Microsoft Research Cambridge\\
\texttt{Sam.Devlin@microsoft.com}\\
\And
Frans A. Oliehoek\\
Interactive Intelligence Group\\
Delft University of Technology\\
\texttt{F.A.Oliehoek@tudelft.nl}\\
\And
Rahul Savani\\
Dept. of Computer Science\\
University of Liverpool\\
\texttt{rahul.savani@liverpool.ac.uk}}
\date{}

\begin{document}
\maketitle

\begin{abstract}
Policy gradient methods have become one of the most popular classes of algorithms for multi-agent reinforcement learning. A key challenge, however, that is not addressed by many of these methods is multi-agent credit assignment: assessing an agent's contribution to the overall performance, which is crucial for learning good policies. We propose a novel algorithm called Dr.Reinforce that explicitly tackles this by combining difference rewards with policy gradients to allow for learning decentralized policies when the reward function is known. By differencing the reward function directly, Dr.Reinforce avoids difficulties associated with learning the Q-function as done by Counterfactual Multiagent Policy Gradients (COMA), a state-of-the-art difference rewards method. For applications where the reward function is unknown, we show the effectiveness of a version of Dr.Reinforce that learns an additional reward network that is used to estimate the difference rewards.
\end{abstract}


\section{Introduction}
Many real-world problems, like air traffic management \citep{airflow}, packet routing in sensor networks \citep{sensors}, and traffic light control \citep{traffic}, can be naturally modelled as \emph{cooperative multi-agent systems} \citep{mas}. Here multiple agents must learn to work together to achieve a common goal. Such problems have commonly been approached with \emph{multi-agent reinforcement learning} (MARL) \citep{marl,planning,cooperative}, including recently with \emph{deep reinforcement learning}. Often in these settings agents have to behave in a \emph{decentralized fashion} \citep{independent}, relying only on local perceptions, due to the prohibitive complexity of a centralized solution or because communication is too expensive \citep{mmdp,decpomdp}.

The paradigm of \emph{centralized training with decentralized execution} (CTDE) \citep{planning,ctde} deals with this: agents use global information during training, but then only rely on local sensing during execution. In such settings, policy gradient methods are amongst the few methods with convergence guarantees \citep{mapg}, and multi-agent policy gradient (MAPG) methods have become one of the most popular approaches for the CTDE paradigm \citep{coma,maddpg}.

However, one key problem that agents face with CTDE that is not directly tackled by many MAPG methods is \emph{multi-agent credit assignment} \citep{local,credit,objective,coin}. With a shared reward signal, an agent cannot readily tell how its own actions affect the overall performance. This can lead to sub-optimal policies even with just a few agents. \emph{Difference rewards} \citep{aristocrat,difference,reward1,reward2} were proposed to tackle this problem: agents learn from a shaped reward that allows them to infer how their actions contributed to the shared reward.

Only one MAPG method has incorporated this idea so far: Counterfactual Multiagent Policy Gradients (COMA) \citep{coma} is a state-of-the-art algorithm that does the differencing with a learned action-value function $Q_{\omega}(s,a)$. However, there are potential disadvantages to this approach: learning the $Q$-function is a difficult problem due to compounding factors of bootstrapping, the moving target problem (as target values used in the update rule change over time), and $Q$'s dependence on the joint actions. This makes the approach difficult to apply with more than a few agents. Moreover, COMA is not exploiting knowledge about the reward function, even though this might be available in many MARL problems.

To overcome these potential difficulties, we take inspiration from \citep{reward2} and incorporate the \emph{differencing of the reward function} into MAPG. Extending the work in \citep{ea} with additional results and analysis, we propose \emph{difference rewards REINFORCE} (Dr.Reinforce), a new MARL algorithm that combines decentralized policies learned with policy gradients with difference rewards that are used to provide gradients with information on each agent's individual contribution to overall performance. Additionally, we provide a version, called Dr.ReinforceR, for settings where the reward function is not known upfront. In contrast to \citep{reward2}, Dr.ReinforceR exploits the CTDE paradigm and learns a centralized reward network to estimate difference rewards. Although the dimensionality of the reward function is the same as the $Q$-function, and similarly depends on joint actions, learning the reward function is a simple regression problem. It does not suffer from the moving target problem, which allows for faster training and improved performance. Our empirical results show that our approaches can significantly outperform other MAPG methods, particularly with more agents.

\section{Background}
Here we introduce some notions about multi-agent systems and policy gradients used to understand the remainder of this work.

\subsection{Multi-Agent Reinforcement Learning}
Our setting can be formalized as a multi-agent Markov decision process (MMDP) \citep{mmdp} $\mathcal{M}=\langle D,S,\{A^i\}_{i=1}^{\vert D\vert },T,R,\gamma\rangle$, where $D=\{1,\ldots,N\}$ is the set of agents; $s\in S$ is the state; $a^i\in A^i$ is the action taken by agent $i$ and $a=\langle a^1,\ldots,a^N\rangle\in\times_{i=1}^{\vert D\vert }A^i=A$ denotes the joint action; $T(s'\vert a,s):S\times A\times S\rightarrow [0,1]$ is the transition function that determines the probability of ending up in state $s'$ from $s$ under joint action $a$; $R(s,a):S\times A\rightarrow \mathbb{R}$ is the shared reward function and $\gamma$ is the discount factor.

Agent $i$ selects actions using a stochastic policy $\pi_{\theta^i}(a^i\vert s):S\times A^i\rightarrow [0,1]$ with parameters $\theta^i$, with $\theta=\langle \theta^1,\ldots,\theta^N\rangle$ and $\pi_{\theta}=\langle\pi_{\theta^1},\ldots,\pi_{\theta^N}\rangle$ denoting the joint parameters and policy respectively. With $r_t$ denoting the reward at time $t$, and expectations taken over sequences of executions, the policy $\pi_{\theta}$ induces the value-functions $V^{\pi_{\theta}}(s_t)=\mathbb{E}_{\pi_{\theta}}\left[\sum_{l=0}^\infty\gamma^lr_{t+l}\vert s_t\right]$ and action-value function $Q^{\pi_{\theta}}(s_t,a_t)=\mathbb{E}_{\pi_{\theta}}\left[\sum_{l=0}^\infty\gamma^lr_{t+l}\vert s_t,a_t\right]$. At each time step $t$, the agents try to maximize the value-function $V^{\pi_{\theta}}(s_t)$.

\subsection{REINFORCE and Actor-Critic}
In single-agent reinforcement learning \citep{rl,survey}, \emph{policy gradient} methods (PG) \citep{pg} aims to maximize the expected value-function $V^{\pi_{\theta}}(s_t)$ by directly optimizing the policy parameters $\theta$. These methods perform gradient ascent in the direction that maximizes the expected parametrized value-function $V(\theta)=\mathbb{E}_{s_0}\left[V^{\pi_{\theta}}(s_0)\right]$. The simplest policy gradient method is REINFORCE \citep{reinforce}, which is a Monte-Carlo algorithm, executing the current policy $\pi_{\theta}$ for an entire episode of $T$ steps and then optimizing it with the following update:

\begin{equation*}
\theta\leftarrow\theta+\alpha\underbrace{\sum_{t=0}^{T-1}\gamma^tG_t\nabla_{\theta}\log\pi_{\theta}(a_t\vert s_t)}_{\hat{g}},
\end{equation*}
where the return $G_t=\sum_{l=0}^{T-t-1}\gamma^lr_{t+l}$ is an unbiased estimate of $V^{\pi_{\theta}}(s_t)$ computed over the episode. This update rule corresponds to performing stochastic gradient ascent \citep{sgd} on $V(\theta)$ because the expectation of the update target is the gradient of the value-function, $\mathbb{E}_{\pi_{\theta}}\left[\hat{g}\right]=\nabla_{\theta}V(\theta)$. Under appropriate choices of step sizes $\alpha$ the method will converge \citep{pg}.

REINFORCE suffers from the high variance of the sampled returns because of the stochasticity of environment and agent policy itself, and thus converges slowly. To reduce such variance, a suitable baseline $b(s)$ can be subtracted from the return $G_t$ \citep{rl}.

Another possibility to overcome such problem are \emph{actor-critic} methods \citep{ac,a3c}, that try to do so by learning an additional component called the critic. The critic is parametrized by $\omega$ and represents either the value or action-value function. It is learned along with the policy $\pi_{\theta}$ to minimize the on-policy \emph{temporal-difference (TD-)error} at each time step $t$, which for a critic that represents the $Q$-function is:

\begin{equation}
\delta_t=r_t+\gamma Q_{\omega}(s_{t+1},a_{t+1})-Q_{\omega}(s_t,a_t).
\label{eq:td}
\end{equation}

The policy is then optimized using the estimates given by the critic:

\begin{equation}
\theta\leftarrow\theta+\alpha\sum_{t=0}^{T-1}Q_{\omega}(s_t,a_t)\nabla_{\theta}\log\pi_{\theta}(a_t\vert s_t).
\label{eq:pg}
\end{equation}

As for REINFORCE, a baseline $b(s)$ can be subtracted from the critic estimate in Equation \eqref{eq:pg} to further reduce variance. If $b(s)=V(s)$, then $A(s,a)=Q_{\omega}(s,a)-V(s)$ is called the \emph{advantage function} and is a used in many actor-critic methods \citep{a3c}.

In cooperative MARL, each agent $i$ can individually learn a decentralized policy by using the \emph{distributed policy gradient} \citep{mapg} update target for $\pi_{\theta^i}$:

\begin{equation}
\theta^i\leftarrow\theta^i+\alpha\underbrace{\sum_{t=0}^{T-1}\gamma^tG_t\nabla_{\theta^i}\log\pi_{\theta^i}(a^i_t\vert s_t)}_{\hat{g}^i},
\label{eq:mapg}
\end{equation}
where $a^i$ is this agent's action and $G_t$ is the return computed with the shared reward and is identical for all agents.

\subsection{Difference Rewards}
In settings where the reward signal is shared, agents cannot easily determine their individual contribution to the reward, a problem known as multi-agent credit assignment. It can be tackled with difference rewards \citep{aristocrat,difference}. Instead of using the shared reward $R(s,a)$, agents compute a shaped reward:

\begin{equation}
\Delta R^i(a^i\vert s,a^{-i})=R(s,a)-R(s,\langle a^{-i},c^i\rangle),
\label{eq:dr}
\end{equation}
where $a^{-i}$ is the joint action all agents except $i$ and $c^i$ is a \emph{default action} for agent $i$ used to replace $a^i$. This way, an agent can assess its own contribution, and therefore each action that improves $\Delta R^i$ also improves the global reward $R(s,a)$ \citep{visualizing}. This however requires access to the complete reward function, or the use of a resettable simulator to estimate $R(s,\langle a^{-i},c^i\rangle)$. Moreover, the choice of the default action can be problematic. The \emph{aristocrat utility} \citep{aristocrat} avoids this choice by marginalizing out an agent by computing its expected contribution to the reward given its current policy $\pi_{\theta^i}$:

\begin{equation}
\Delta R^i(a^i\vert s,a^{-i})=R(s,a)-\mathbb{E}_{b^i\sim\pi_{\theta^i}}\left[R(s,\langle a^{-i},b^i\rangle)\right].
\label{eq:aristocrat}
\end{equation}

The work of \citep{reward2} learns a local approximation of the reward function $R_{\psi^i}(s,a^i)$ for each agent $i$, and uses it to compute the difference rewards of Equation \eqref{eq:dr}, by fixing a default action $c^i$, as:

\begin{equation*}
\Delta R_{\psi^i}^i(a^i\vert s)=R(s,a)-R_{\psi^i}(s,c^i).
\end{equation*}

Counterfactual Multiagent Policy Gradients (COMA) \citep{coma} is a state-of-the-art deep MAPG algorithm that adapts difference rewards and aristocrat utility to use the $Q$-function, approximated by a centralized critic $Q_{\omega}(s,a)$ learned under the CTDE paradigm (as the algorithm is designed for general partially observable multi-agent domains \citep{decpomdp}, where agents cannot access the environment state $s_t$), by providing the policy gradients of the agents with a counterfactual advantage function:

\begin{equation*}
A^i(s,a)=Q_{\omega}(s,a)-\sum_{c^i\in A^i}\pi_{\theta^i}(c^i\vert h^i_t)Q_{\omega}(s,\langle a^{-i},c^i\rangle).
\end{equation*}

\section{Difference Rewards Policy Gradients}
COMA learns a centralized action-value function critic $Q_{\omega}(s,a)$ to do the differencing and drive agents' policy gradients. However, learning such a critic using the TD-error in Equation \eqref{eq:td} presents a series of challenges that may dramatically hinder final performance if they are not carefully tackled. The $Q$-value updates rely on bootstrapping that can lead to inaccurate updates. Moreover, the target values for these updates are constantly changing because the other estimates used to compute them are also updated, leading to a moving target problem. This is exacerbated when function approximation is used, as these estimates can be indirectly modified by the updates of other $Q$-values. Target networks are used to try and tackle this problem \citep{dqn}, but these require careful tuning of additional parameters and may slow down convergence with more agents.

Our proposed algorithm, named Dr.Reinforce, combines the REINFORCE \citep{reinforce} policy gradient method with a difference rewards mechanism to deal with credit assignment in cooperative multi-agent systems, thus avoiding the need of learning a critic.

\subsection{Dr.Reinforce}
If the reward function $R(s,a)$ is known, we can directly use difference rewards with policy gradients. We define the \emph{difference return} $\Delta G^i_t$ for agent $i$ as the discounted sum of the difference rewards $\Delta R^i(a_t^i\vert s_t,a_t^{-i})$ from time step $t$ onward as:

\begin{equation}
\Delta G^i_t(a_{t:T}^i\vert s_{t:T},a_{t:T}^{-i})\triangleq\sum_{l=0}^{T-t-1}\gamma^l\Delta R^i(a_{t+l}^i\vert s_{t+l},a_{t+l}^{-i}),
\label{eq:diff_return}
\end{equation}
where $T$ is the length of the sampled trajectory and $\Delta R^i(a_t^i\vert s_t,a_t^{-i})$ is the difference rewards for agent $i$, computed using the aristocrat utility \citep{aristocrat} as in Equation \eqref{eq:aristocrat}. Please note that the subscript $t:T$ in our notation is a shorthand used to identify the sequence of values of given quantity from time step $t$ up to (but not including) time step $T$.

To learn the decentralized policies $\pi_{\theta}$, we follow a modified version of the distributed policy gradients in Equation \eqref{eq:mapg} that uses our difference return, optimizing each policy by using the update target:

\begin{equation}
\theta^i\leftarrow\theta^i+\alpha\underbrace{\sum_{t=0}^{T-1}\gamma^t\Delta G^i_t(a_{t:T}^i\vert s_{t:T},a_{t:T}^{-i})\nabla_{\theta^i}\log\pi_{\theta^i}(a^i_t\vert s_t)}_{g^{DR,i}},
\label{eq:drpg}
\end{equation}
where $\Delta G^i_t$ is the difference return defined in Equation \eqref{eq:diff_return}. This way, each policy is guided by an update that takes into account its individual contribution to the shared reward, and an agent thus takes into account the real value of its own actions. We expect this signal to drive the policies towards regions in which individual contributions are higher, and thus also the shared reward, since a sequence of actions improving $\Delta G^i_t$ also improves the global return \citep{visualizing}.

\subsection{Online Reward Estimation}
In many settings, complete access to the reward function to compute the difference rewards is not available. Thus, we propose Dr.ReinforceR, which is similar to Dr.Reinforce but additionally learns a \emph{centralized reward network} $R_{\psi}$, with parameters $\psi$, that is used to estimate the value $R(s,\langle a^i,a^{-i}\rangle)$ for every local action $a^i\in A^i$ for agent $i$. Following the CTDE paradigm, this centralized network is only used during training to provide policies with learning signals, and is not needed during execution, when only the decentralized policies are used. The reward network receives as input the environment state $s_t$ and the joint action of the agents $a_t$ at time $t$, and is trained to reproduce the corresponding reward value $r_t\sim R(s_t,a_t)$ by minimizing a standard MSE regression loss:

\begin{equation}
\mathcal{L}_t(\psi)=\frac{1}{2}\left(r_t-R_{\psi}(s_t,a_t)\right)^2.
\label{eq:r}
\end{equation}

Although the dimensionality of the function $R(s,a)$ that we are learning with the reward network is the same as that of $Q(s,a)$ learned by the COMA critic, growing exponentially with the number of agents as both depend of the joint action $a\in A=\times_{i=1}^{\vert D\vert }A^i$, learning $R_{\psi}$ is a regression problem that does not involve bootstrapping or moving targets, thus avoiding many of the problems faced with an action-value function critic. Moreover, alternative representations of the reward function can be used to further improve learning speed and accuracy, e.g., by using factorizations \citep{factored}.

We can now use the learned $R_{\psi}$ to compute the difference rewards $\Delta R_{\psi}^i$ using the aristocrat utility \citep{aristocrat} as:

\begin{equation}
\Delta R_{\psi}^i(a^i_t\vert s_t,a_t^{-i})\triangleq r_t-\sum_{c^i\in A^i}\pi_{\theta^i}(c^i\vert s_t)R_{\psi}(s_t,\langle c^i,a^{-i}_t\rangle).
\label{eq:ar_learned}
\end{equation}

The second term of the r.h.s.\ of Equation \eqref{eq:ar_learned} can be estimated with a number of network evaluations that is linear in the size of the local action set $A^i$, as the actions of the other agents $a^{-i}_t$ remains fixed, avoiding an exponential cost.

We now redefine the difference return $\Delta G^i_t$ from Equation \eqref{eq:diff_return} as the discounted sum of the estimated difference rewards $\Delta R_{\psi}^i(a_{t+l}^i\vert s_{t+l},a_{t+l}^{-i})$:

\begin{equation}
\Delta G^i_t(a_{t:T}^i\vert s_{t:T},a_{t:T}^{-i})\triangleq\sum_{l=0}^{T-t-1}\gamma^l\Delta R_{\psi}^i(a_{t+l}^i\vert s_{t+l},a_{t+l}^{-i}).
\label{eq:diff_return_learned}
\end{equation}

\subsection{Theoretical Results}
\label{sec:theory}
\allowdisplaybreaks
REINFORCE \citep{reinforce} suffers from high variance of gradients estimates because of sample estimation of the return. This can be accentuated in the multi-agent setting. Using an unbiased baseline is crucial to reducing this variance and improving learning \citep{variance,rl}. Here we address these concerns by showing that using difference rewards in policy gradient methods corresponds to subtracting an unbiased baseline from the policy gradient of each individual agent. Since this unbiased baseline does not alter the expected value of the update targets, applying difference rewards policy gradients to a common-reward MARL problem turns out to be same in expectation as using distributed policy gradients update targets. Such gradients' updates have been shown to be equivalent to those of a joint gradient \citep{mapg}, which under some technical conditions is known to converge to a local optimum \citep{pg,ac}.

\begin{lemma}
In a MMDP, using difference return $\Delta G^i_t(a_{t:T}^i\vert s_{t:T},a_{t:T}^{-i})$ as the learning signal for policy gradients in Equation \eqref{eq:drpg} is equivalent to subtracting an unbiased baseline $B^i(s_{t:T},a^{-i}_{t:T})$ from the distributed policy gradients in Equation \eqref{eq:mapg}.
\label{thm:lemma}
\end{lemma}

\begin{proof}
We start by rewriting $\Delta G^i_t(a_{t:T}^i\vert s_{t:T},a_{t:T}^{-i})$ from Equation \eqref{eq:diff_return} as:

\begin{equation}
\Delta G^i_t(a_{t:T}^i\vert s_{t:T},a_{t:T}^{-i})=\sum_{l=0}^{T-t-1}\gamma^lr_{t+l}-\sum_{l=0}^{T-t-1}\gamma^l\sum_{c^i\in A^i}\pi_{\theta^i}(c^i\vert h^i_{t+l})R(s_{t+l},\langle c^i,a^{-i}_{t+l}\rangle).
\label{eq:diff_return_div}
\end{equation}

Note that the first term on the r.h.s. of Equation \eqref{eq:diff_return_div} is the return $G_t$ used in Equation \eqref{eq:mapg}. We then define the second term on the r.h.s. of Equation \eqref{eq:diff_return_div} as the baseline $B^i(s_{t:T},a^{-i}_{t:T})$:

\begin{equation}
B^i(s_{t:T},a^{-i}_{t:T})=\sum_{l=0}^{T-t-1}\gamma^l\sum_{c^i\in A^i}\pi_{\theta^i}(c^i\vert s_{t+l})\cdot R(s_{t+l},\langle c^i,a^{-i}_{t+l}\rangle).
\label{eq:baseline_fo}
\end{equation}

We can thus rewrite the total expected update target for agent $i$ as:

\begin{align}
\mathbb{E}_{\pi_{\theta}}\left[\hat{g}^{DR,i}\right] & =\mathbb{E}_{\pi_{\theta}}\left[\sum_{t=0}^{T-1}\left(\nabla_{\theta^i}\log\pi_{\theta^i}(a^i_t\vert s_t)\right)\Delta G^i_t(a_{t:T}^i\vert s_{t:T},a_{t:T}^{-i})\right] \nonumber\\
& =\mathbb{E}_{\pi_{\theta}}\left[\sum_{t=0}^{T-1}\left(\nabla_{\theta^i}\log\pi_{\theta^i}(a^i_t\vert s_t)\right)\left(G_t-B^i(s_{t:T},a^{-i}_{t:T})\right)\right] \nonumber\\
& \text{(by definition of }\Delta G^i_t) \nonumber\\
& =\mathbb{E}_{\pi_{\theta}}\left[\sum_{t=0}^{T-1}\left(\nabla_{\theta^i}\log\pi_{\theta^i}(a^i_t\vert s_t)\right)G_t\right. \nonumber\\
& \left.-\left(\nabla_{\theta^i}\log\pi_{\theta^i}(a^i_t\vert s_t)\right)B^i(s_{t:T},a^{-i}_{t:T})\right] \nonumber\\
& \text{(distributing the product)} \nonumber\\
& =\mathbb{E}_{\pi_{\theta}}\left[\sum_{t=0}^{T-1}\left(\nabla_{\theta^i}\log\pi_{\theta^i}(a^i_t\vert s_t)\right)G_t\right] \nonumber\\
& -\mathbb{E}_{\pi_{\theta}}\left[\sum_{t=0}^{T-1}\left(\nabla_{\theta^i}\log\pi_{\theta^i}(a^i_t\vert s_t)\right)B^i(s_{t:T},a^{-i}_{t:T})\right] \nonumber\\
& \text{(by linearity of the expectation)} \nonumber\\
& =\mathbb{E}_{\pi_{\theta}}\left[\hat{g}^{PG,i}\right]+\mathbb{E}_{\pi_{\theta}}\left[\hat{g}^{B,i}\right]. \label{eq:grad}
\end{align}

We have to show that the baseline is unbiased, and so the expected value of its update $\mathbb{E}_{\pi_{\theta}}\left[\hat{g}^{B,i}\right]$ with respect to the policy $\pi_{\theta}$ is $0$. Let $P^{\pi_{\theta}}_t(s_t)=\sum_{s_{t-1}\in S}P^{\pi_{\theta}}_{t-1}(s_{t-1})\sum_{a_{t-1}\in A}\pi_{\theta}(a_{t-1}\vert s_{t-1})\; T(s_t\vert a_{t-1},s_{t-1})$ be the probability of the state at time step $t$ to be $s_t$ under the joint policy $\pi_{\theta}$ (with $P^{\pi_{\theta}}_0(s_0)=\rho(s_0)$ and $\rho$ is the initial state distribution), we have:

\begin{align}
\mathbb{E}_{\pi_{\theta}}\left[\hat{g}^{B,i}\right] & \;\triangleq-\mathbb{E}_{\pi_{\theta}}\left[\sum_{t=0}^{T-1}\left(\nabla_{\theta^i}\log\pi_{\theta^i}(a^i_t\vert s_t)\right)B^i(s_{t:T},a^{-i}_{t:T})\right] \nonumber\\
& \;=-\sum_{t=0}^{T-1}\sum_{s_t\in S}P^{\pi_{\theta}}_t(s_t)\sum_{a^{-i}_t\in A^{-i}}\pi_{\theta^{-i}}(a^{-i}_t\vert s_t)\sum_{a^i_t\in A^i}\pi_{\theta^i}(a^i_t\vert s_t) \nonumber\\
& \left(\nabla_{\theta^i}\log\pi_{\theta^i}(a^i_t\vert s_t)\right)\sum_{s_{t+1:T},a_{t+1:T}} \nonumber\\
& \prod_{l=1}^{T-t-1}T(s_{t+l}\vert a_{t+l-1},s_{t+l-1})\cdot\pi_{\theta}(a_{t+l}\vert s_{t+l})\; B^i(s_{t:T},a^{-i}_{t:T}) \nonumber\\
& \text{(by expanding the expectation)} \nonumber\\
& \;=-\sum_{t=0}^{T-1}\sum_{s_t\in S}P^{\pi_{\theta}}_t(s_t)\sum_{a^{-i}_t\in A^{-i}}\pi_{\theta^{-i}}(a^{-i}_t\vert s_t)\sum_{a^i_t\in A^i}\left(\nabla_{\theta^i}\pi_{\theta^i}(a^i_t\vert s_t)\right) \nonumber\\
& \sum_{s_{t+1:T},a_{t+1:T}}\prod_{l=1}^{T-t-1}T(s_{t+l}\vert a_{t+l-1},s_{t+l-1})\cdot\pi_{\theta}(a_{t+l}\vert s_{t+l})\; B^i(s_{t:T},a^{-i}_{t:T}) \nonumber\\
& \text{(by applying the inverse log trick)} \nonumber\\
& \;=-\sum_{t=0}^{T-1}\sum_{s_t\in S}P^{\pi_{\theta}}_t(s_t)\sum_{a^{-i}_t\in A^{-i}}\pi_{\theta^{-i}}(a^{-i}_t\vert s_t)\left(\nabla_{\theta^i}\sum_{a^i_t\in A^i}\pi_{\theta^i}(a^i_t\vert s_t)\right) \nonumber\\
& \sum_{s_{t+1:T},a_{t+1:T}}\prod_{l=1}^{T-t-1}T(s_{t+l}\vert a_{t+l-1},s_{t+l-1})\cdot\pi_{\theta}(a_{t+l}\vert s_{t+l})\; B^i(s_{t:T},a^{-i}_{t:T}) \nonumber\\
& \text{(by moving the gradient outside the policy sum)} \nonumber\\
& \;=-\sum_{t=0}^{T-1}\sum_{s_t\in S}P^{\pi_{\theta}}_t(s_t)\sum_{a^{-i}_t\in A^{-i}}\pi_{\theta^{-i}}(a^{-i}_t\vert s_t)\;\nabla_{\theta^i}1 \nonumber\\
& \sum_{s_{t+1:T},a_{t+1:T}}\prod_{l=1}^{T-t-1}T(s_{t+l}\vert a_{t+l-1},s_{t+l-1})\cdot\pi_{\theta}(a_{t+l}\vert s_{t+l})\; B^i(s_{t:T},a^{-i}_{t:T}) \nonumber\\
& \text{(policy probabilities sum up to 1)} \nonumber\\
&\;=0. \label{eq:derivation} 
\end{align}

Therefore, using the baseline in Equation \eqref{eq:baseline_fo} reduces the variance of the updates \citep{variance} but does not change their expected value, as it is unbiased and its expected update target $\mathbb{E}_{\pi_{\theta}}\left[\hat{g}^{B,i}\right]=0$.
\end{proof}

\begin{corollary}
Using the estimated reward network $R_{\psi}$ to compute the baseline in Equation \eqref{eq:baseline_fo} still results in an unbiased baseline.
\end{corollary}

\begin{proof}
We rewrite $\Delta G^i_t(a_{t:T}^i\vert s_{t:T},a_{t:T}^{-i})$ from Equation \eqref{eq:diff_return_learned} as:

\begin{equation}
\Delta G^i_t(a_{t:T}^i\vert s_{t:T},a_{t:T}^{-i})=\sum_{l=0}^{T-t-1}\gamma^lr_{t+l}-\sum_{l=0}^{T-t-1}\gamma^l\sum_{c^i\in A^i}\pi_{\theta^i}(c^i\vert s_{t+l})R_{\psi}(s_{t+l},\langle c^i,a^{-i}_{t+l}\rangle),
\label{eq:diff_return_learned_div}
\end{equation}
for which we define the second term on the r.h.s. of Equation \eqref{eq:diff_return_learned_div} as the baseline $B_{\psi}^i(s_{t:T},a^{-i}_{t:T})$:

\begin{equation*}
B_{\psi}^i(s_{t:T},a^{-i}_{t:T})=\sum_{l=0}^{T-t-1}\gamma^l\sum_{c^i\in A^i}\pi_{\theta^i}(c^i\vert s_{t+l})\cdot R_{\psi}(s_{t+l},\langle c^i,a^{-i}_{t+l}\rangle).
\end{equation*}

We observe that the derivation of Equation \eqref{eq:derivation} still holds, as it is not altered by the use of the reward network $R_{\psi}$ rather than the true reward function $R(s,a)$. Therefore, the baseline $B_{\psi}^i(s_{t:T},a^{-i}_{t:T})$ is again unbiased and does not alter the expected value of the updates.
\end{proof}

\begin{theorem}
In a MMDP with shared rewards, given the conditions on function approximation detailed in \citep{pg}, using Dr.Reinforce update target as in Equation \eqref{eq:drpg}, the series of parameters $\{\theta_t=\langle\theta^1_t,\ldots,\theta^N_t\rangle\}_{t=0}^k$ converges in the limit such that the corresponding joint policy $\pi_{\theta_t}$ is a local optimum:
$$\lim_{k\rightarrow\infty}\inf_{\{\theta_t\}_{t=0}^k}\vert \vert \hat{g}^{DR}\vert \vert =0\qquad w.p.\;1.$$
\end{theorem}

\begin{proof}
To prove convergence, we have to show that:

\begin{equation*}
\mathbb{E}_{\pi_{\theta_t}}\left[\hat{g}^{DR}\right]=\mathbb{E}_{\pi_{\theta_t}}\left[\sum_{i=0}^N\hat{g}^{DR,i}\right]=\nabla_{\theta_t}V(\theta_t).
\end{equation*}

We can rewrite the total expected update target as:

\begin{equation*}
\mathbb{E}_{\pi_{\theta_t}}\left[\hat{g}^{DR,i}\right]=\mathbb{E}_{\pi_{\theta_t}}\left[\hat{g}^{PG,i}\right]+\mathbb{E}_{\pi_{\theta_t}}\left[\hat{g}^{B,i}\right]
\end{equation*}
as in Equation \eqref{eq:grad}, and by Lemma \ref{thm:lemma} we have that $\mathbb{E}_{\pi_{\theta_t}}\left[\hat{g}^{B,i}\right]=0$. Therefore, the overall expected update $\mathbb{E}_{\pi_{\theta_t}}\left[\hat{g}^{DR,i}\right]$ for agent $i$ reduces to $\mathbb{E}_{\pi_{\theta_t}}\left[\hat{g}^{PG,i}\right]$, that is equal to the distributed policy gradient update target in Equation \eqref{eq:mapg}. These updates for all the agents has been proved to be equal to these of a centralized policy gradients agent $\mathbb{E}_{\pi_{\theta_t}}\left[\hat{g}^{PG}\right]$ by Theorem 1 in \citep{mapg}, and therefore converge to a local optimum of $\nabla_{\theta_t}V(\theta_t)$ by Theorem 3 in \citep{pg}.
\end{proof}

\section{Experiments}
We are interested in investigating the following questions:

\begin{enumerate}
\item How does Dr.Reinforce compare to existing approaches?
\item How does the use of a learned reward network $R_{\psi}$ instead of a known reward function affect performance?
\item Is learning the $Q$-function (as in COMA) more difficult than learning the reward function $R(s,a)$ (as in Dr.ReinforceR)?
\end{enumerate}

To investigate these questions, we tested our methods on two gridworld environments with shared reward: the multi-rover domain, an established multi-agent cooperative domain \citep{reward1}, in which agents have to spread across a series of landmarks, and a variant of the classical predator-prey problem with a randomly moving prey \citep{qlearn}.

\begin{figure}[!t]
\centering
\subfloat[Multi-Rover]{
\includegraphics[width=0.3\textwidth]{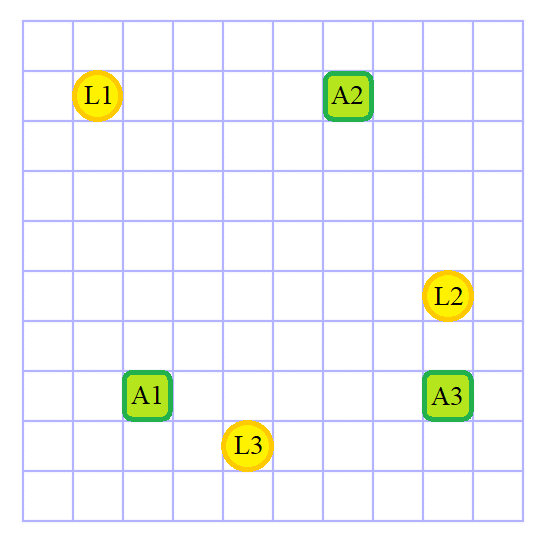}
}
\subfloat[Predator-Prey]{
\includegraphics[width=0.3\textwidth]{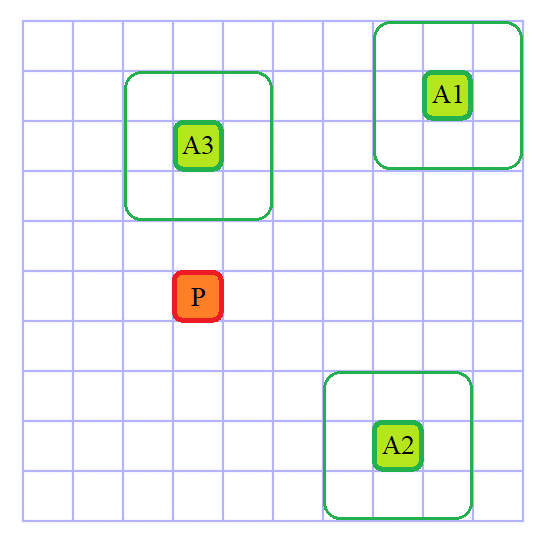}
}
\caption{Schematic representation of the two gridworld domains. Agents are green, landmarks are yellow, and the prey is red.}
\end{figure}

\subsection{Comparison to Baselines}
\label{sec:comparisons}
We compare to a range of other policy gradient methods: independent learners using REINFORCE to assess the benefits of using a difference rewards mechanism, labelled PG. We also compare against a standard actor-critic algorithm \citep{ac} with decentralized actors and a centralized action-value function critic to show that our improvements are not only due to the centralized information provided to the agents during training, denoted as CentralQ here. Our main comparison is with COMA \citep{coma}, a state-of-the-art difference rewards method using the $Q$-function for computing the differences. Finally, we compare against the algorithm proposed in \citep{reward2}, to show the benefit of learning a centralized reward network to estimate the difference rewards in Dr.ReinforceR. Indeed, this algorithm learns an individual approximation of the reward function $R_{\psi^i}(s,a^i)$  for each agent $i$, and uses this in estimating the difference rewards as in Equation \ref{eq:dr} to learn the agents' policies. We adapted this method to use policy gradients instead of evolutionary algorithms to optimise the policies to not conflate the comparisons with the choice of a policy optimizer where possible, and only focus on the effect of using difference rewards during learning. Additionally, the multi-agent A\textsuperscript{*} (MAA\textsuperscript{*}) exact planning algorithm \citep{maastar,decpomdp} has been applied to the smaller instances of the two problems with only $N=3$ agents, as an upper bound for assessing the overall performance of the investigated learning algorithms. Because of the exponentially many joint actions to expand at each state, it has not been possible to apply such an algorithm to larger instances.

\begin{figure}[!t]
\centering
\subfloat[Multi-Rover, $N=3$]{
\includegraphics[width=0.48\textwidth]{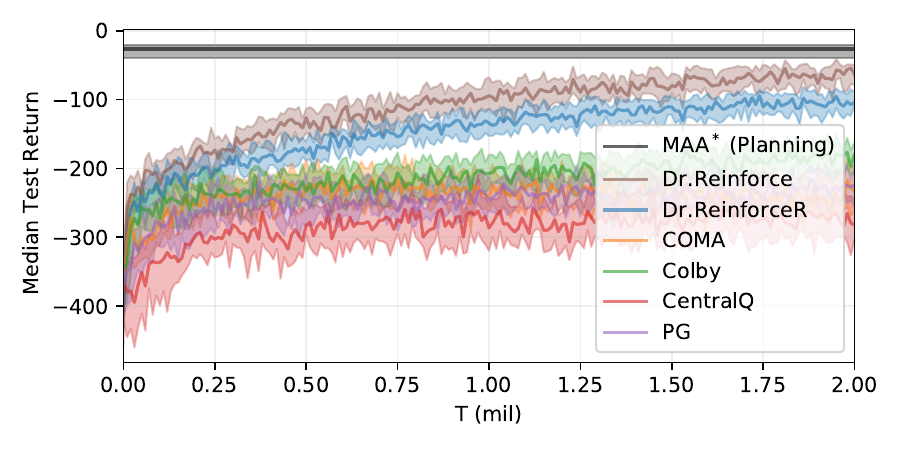}
}
\subfloat[Predator-Prey, $N=3$]{
\includegraphics[width=0.48\textwidth]{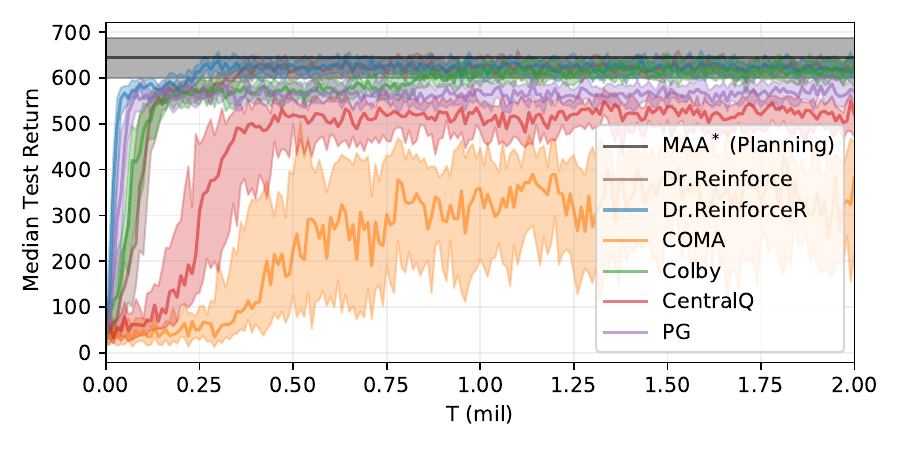}
}

\subfloat[Multi-Rover, $N=5$]{
\includegraphics[width=0.48\textwidth]{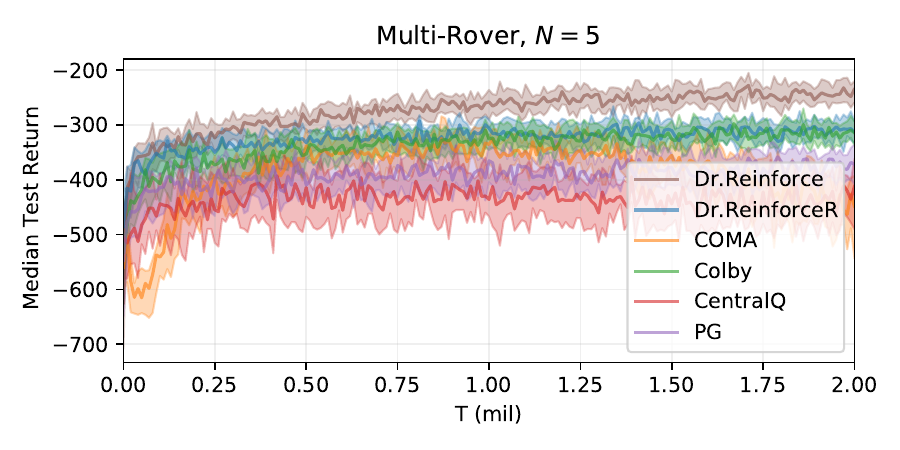}
}
\subfloat[Predator-Prey, $N=5$]{
\includegraphics[width=0.48\textwidth]{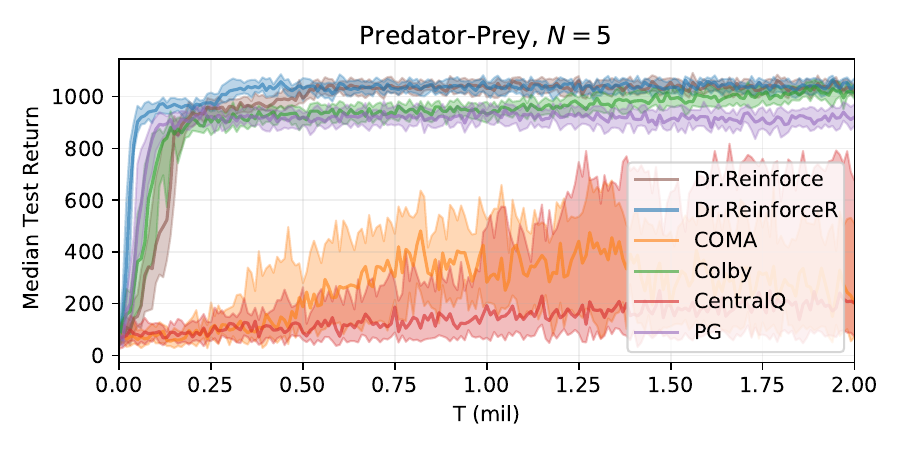}
}

\subfloat[Multi-Rover, $N=8$]{
\includegraphics[width=0.48\textwidth]{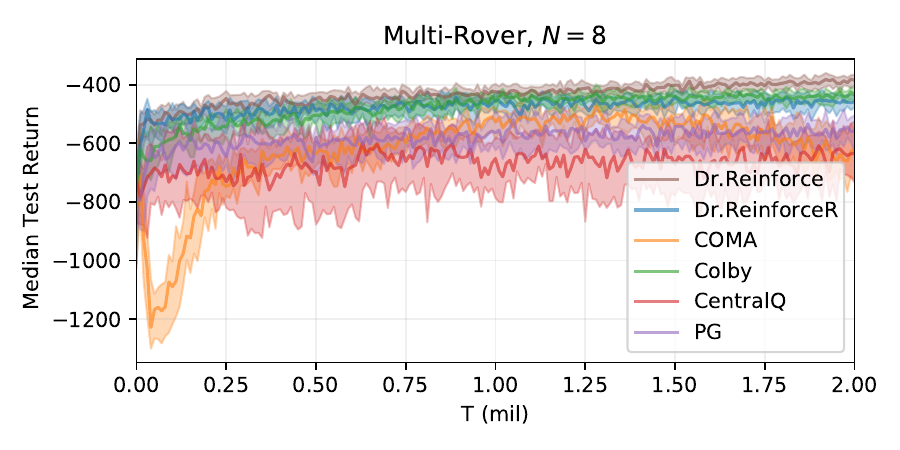}
}
\subfloat[Predator-Prey, $N=8$]{
\includegraphics[width=0.48\textwidth]{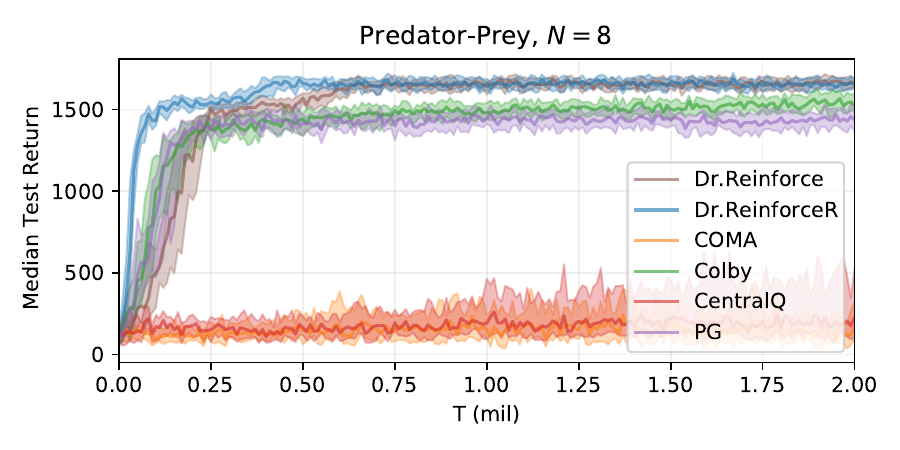}
}
\caption{Training curves on the multi-rover domain (left) and the predator-prey problem (right), showing the median reward and $25-75\%$ percentiles across seeds.}
\label{fig:gridworld}
\end{figure}

\subsubsection{Multi-Rover Domain}
In this domain, a team of $N$ agents is placed into a $10\times 10$ gridworld with a set of $N$ landmarks. The aim of the team is to spread over all the landmarks and cover them (which agent covers which landmark is not important): the reward received by the team depends on the distance of each landmark to its closest agent, and a penalty if two agents collide (reach the same cell simultaneously) during their movements is also applied. Each agent observes its relative position with respect to all the other agents and landmarks, and can move in the four cardinal directions or stand still. Figure \ref{fig:gridworld} (left) reports the median performance and $25-75\%$ percentiles (shaded area in the plots) across $30$ independent runs obtained by the compared methods on a team of increasing size, to investigate scaling to larger multi-agent systems.

It can be observed that both Dr.Reinforce and Dr.ReinforceR are always outperforming all of the other compared baselines on this domain. Also, Dr.ReinforceR is generally matching the upper bound given by Dr.Reinforce (that represents a limit case when the centralized reward network $R_{\psi}$ has perfectly converged to the true reward function). However, the wide gap between these two algorithms and the other baselines when $N=3$ reduces when more agents are introduced in the system, possibly pointing out that also these methods start to struggle in achieving optimal and coordinated behaviours on larger instances of this domain. When more agents are present, the gridworld becomes quite crowded: an explanation for this loss in performance is that the difference rewards signal pushes each agent towards the landmark that is furthest from all of the agents, and thus contributing the most to the negative reward value, in an attempt to mitigate this problem, but letting another landmark increase its negative contribution in turn. Coordination is key to efficiently solve this domain, and achieving such coordination may be difficult in larger settings.

Moreover, even if the reward network learns a good representation, the synergy between this and the agents' policies has to be carefully considered: the reward network has to converge properly before the policies got stuck into a local optimum, or it could be the case that these will not be able to escape it even if the gradients signals are then accurate enough. However, the simpler learning problem used to provide signals to the agents' policies, as opposed to the very complex learning of the action-value function critic used by COMA, proves effective in speeding up learning and achieve higher returns, even in difficult settings with many agents where all the other policy gradient methods seem to fail as well. Computing the difference rewards requires very accurate reward estimates, so if the reward network do not exhibit appropriate generalization capabilities it may end up overfitting on the reward values encountered during training but not being able to give correct predictions beyond those. It is true however that also difference rewards methods using the action-value function like COMA have the same requirements.

\subsubsection{Predator-Prey}
In this version of the classical predator-prey problem, a team of $N$ predators has to pursue a single prey for as long as possible in a $10\times 10$ gridworld. Each predator has got a range of sight of one cell in each direction from its current position: if the prey is into this range, the whole team receives a positive reward bonus, otherwise they do not receive any reward. Each agent observes its relative position with respect to the other agents and the prey itself and can move in the four cardinal directions or stand still. The prey selects actions uniformly at random from the same set of actions available to the agents. Figures \ref{fig:gridworld} (right) shows median results and $25-75\%$ percentiles across $30$ independent runs with teams comprising an increasing number of predators.

Also in this environment, Dr.ReinforceR is outperforming all the other compared methods, achieving performance that is equal or close to these of the Dr.Reinforce upper bound (of which the former is an approximated version). On one hand, some of the other baselines are also performing well: PG and Colby are almost performing on-par with the two above algorithms, even on larger instances of the problem. This is probably due to the less strict coordination requirements of the predator-prey problem compared to the previous multi-rover domain: each agent is independently contributing towards the common goal, and thus simply needs to optimize its own behaviour by learning how to reach and stay on the prey in order to improve global performances. 

On the other hand, COMA is performing extremely poorly, being outperformed even by the simple CentralQ (that has slowly learned something in the simpler case with $N=3$). This points out how accurately learning an optimal $Q$-function may be problematic in many settings, even more so on a sparse setting such as this, in which the agents are only perceiving rewards if some of them are effectively on the prey. If the $Q$-function converges to a sub-optimal solution and keeps pushing the agents towards a local optimum, the policies may struggle to escape from it afterwards and in turn push the action-value function towards a worst approximation. Moreover, to compute the counterfactual baseline in COMA, estimates of $Q$-values need to be accurate even on state-action pairs that the policies do not visit often, further exacerbating this problem. From this side, learning the reward function to compute the difference rewards is an easier learning problem, cast as a regression task and not involving bootstrapped estimates or a moving target, and thus can improve policy gradient performance providing them with better learning signals in achieving high return behaviours with no further drawback.

\subsection{Analysis}
The results of the proposed experiments show the benefits of learning the reward function over the more complex $Q$-function, leading to faster policy training and improved final performances, but also that this is not always an easy task and it can present issues on its own that can hinder the learning of an optimal joint policy. Indeed, although not suffering from the moving target problem and no bootstrapping is involved, learning the reward function online together with the policies of the agents can lead to biases of the learned function due to the agents behaviours. These biases could push the training samples towards a specific region of the true reward function, hindering the generalization capacity of the learned reward network and in turn leading to worst learning signal for the policies themselves, that can get stuck into a sub-optimal region. Similarly, this problem can appear when a centralized action-value critic is used to drive the policy gradients.

To investigate the claimed benefits of learning the reward function rather the $Q$-function, let now analyse the accuracy of the learned representations on the two proposed gridworld domains by sampling a set of different trajectories from the execution of the corresponding policies and comparing the predicted values from the reward network $R_{\psi}(s,a)$ of Dr.ReinforceR and the $Q_{\omega}(s,a)$ critic from COMA to the real ground-truth values of the reward function and the $Q$-function respectively. This has been called the \emph{on-policy dataset}, representing how correctly can the reward network and the critic represent the values of state-action pairs encountered during their training phase. Moreover, both Dr.ReinforceR and COMA rely on a difference rewards mechanism and thus need to estimate values for state-action pairs that are only encountered infrequently (or not at all) in order to compute correct values to drive the policy gradients. To investigate the generalization performances of the learned representations, let also analyse the prediction error on a \emph{off-policy dataset}, by sampling uniformly across the entire action-state space $S\times A$ and again comparing the predicted values from the learned reward function $R_{\psi}(s,a)$ of Dr.ReinforceR and the $Q_{\omega}(s,a)$ critic from COMA to their corresponding ground-truth values. Please note that, not knowing the true $Q$-function for the proposed problems to compare against, these have been approximated that via $100$ rollouts sampled starting from the current state-action sample and following the corresponding learned policies afterwards. Figure \ref{fig:stats} shows the mean and standard deviation of the prediction error (PE) distribution of these networks. All the prediction errors have been normalized by the value of $r_{max}-r_{min}$ (respectively $q_{max}-q_{min}$ for COMA critic) for each environment and number of agents individually, so that the resulting values are comparable across the two different methodologies and across different settings. It is to note that, although normalized, the errors may be higher than the normalization range itself, and thus exceed the value of $1$ (as it is the case with the errors of COMA critic on the multi-rover domain).

\begin{figure}[t]
\centering
\includegraphics[width=0.7\textwidth,clip]{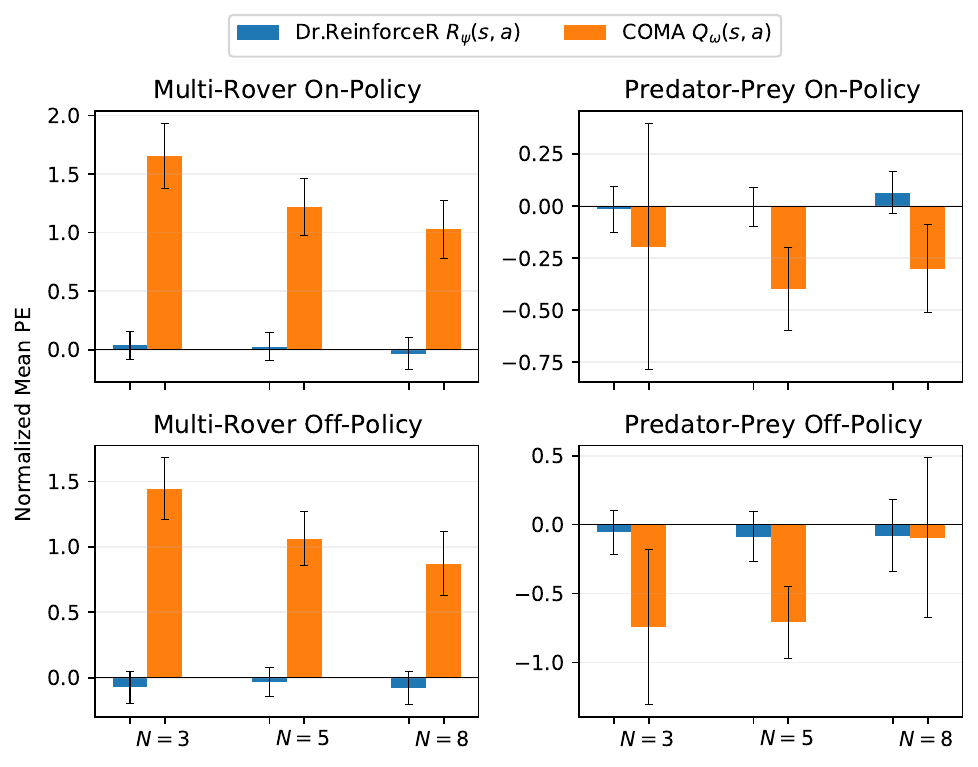}
\caption{Normalized mean prediction error and standard deviation for Dr.ReinforceR reward network $R_{\psi}$ and COMA critic $Q_{\omega}$ on the on-policy dataset (first row) and the off-policy dataset (second row), for the two environments.}
\label{fig:stats}
\end{figure}

These plots give us some insights on the performance reported in Section \ref{sec:comparisons}. Dr.ReinforceR is in general achieving improved performances with respect to the compared baselines, and the low prediction error of its reward network on the two problems may be an explanation for this: with correct value estimates, the learning signals provided to the policy gradients are better in turn, and thus lead to higher-return behaviours. Also the variance is low, meaning that most of the sampled values are consistently predicted correctly and the network exhibits good generalization performances across the increasing number of agents on both datasets. This generalization capacity of the learned approximation also explains why Dr.ReinforceR is in general matching the Dr.Reinforce upper bound: the difference rewards mechanism requires multiple predictions to compute the agents' signals and, if these are not accurate enough, the resulting values may be completely wrong and push the agents towards sub-optimal policies in turn.

The prediction errors for COMA action-value critic instead are higher, especially on the multi-rover domain, where the errors do not scale so gracefully in the number of agents even on the on-policy dataset. It can be observed that the critic network is biased towards overestimating most of the samples for the multi-rover domain, while instead underestimates them for predator-prey (especially more so on the off-policy dataset, where non-encountered state-action pairs may be sampled), thus resulting in bad estimations of the counterfactual baseline. On the predator-prey environment, it seems that COMA critic quickly overfits to the $Q$-function of a sub-optimal joint policy, resulting in a very low prediction error on the off-policy dataset when the number of agents increases (and most of the samples indeed lead to no rewards trajectories), that does not seem able to give good signals to the agents' policies and leads them to get stuck into this poor local optimum in turn. These results can also explain why COMA is performing worse than CentralQ on this domain: if the critic is not accurate or is representing the value of a poor policy (as it can be hypothesized for the above results), COMA requirement of more estimations from it in order to compute the counterfactual baseline only exacerbates this problem and further hinders the final performance.

Finally, the effect of noise on computation of the difference rewards are investigated. Generally, an accurate reward value for every agent's action is needed to compute correct difference rewards. The reward network $R_{\psi}$ is an approximation of the true reward function $R(s,a)$ and can therefore give noisy estimates that could dramatically affect the resulting computation. To investigate this, noise sampled from different processes is added to the reward values of the agent's different actions that are obtained from the environment. These are used to compute the baseline (the second term of the r.h.s. in Equation \ref{eq:aristocrat}, as this is the only term for which $R_{\psi}$ is used in Equation \ref{eq:ar_learned}), and the resulting difference rewards are compared with the true ones for a generic agent $i$ under a uniform policy $\pi_{\theta^i}(a^i\vert s)=\frac{1}{\vert A^i\vert }$. Figure \ref{fig:noise} reports the mean value and variance over $1000$ noise samples of a set of sampled state-action (SA) pairs from the reward function of the two investigated domains with $N=3$ agents.

\begin{figure}[t]
\centering
\includegraphics[width=0.7\textwidth,clip]{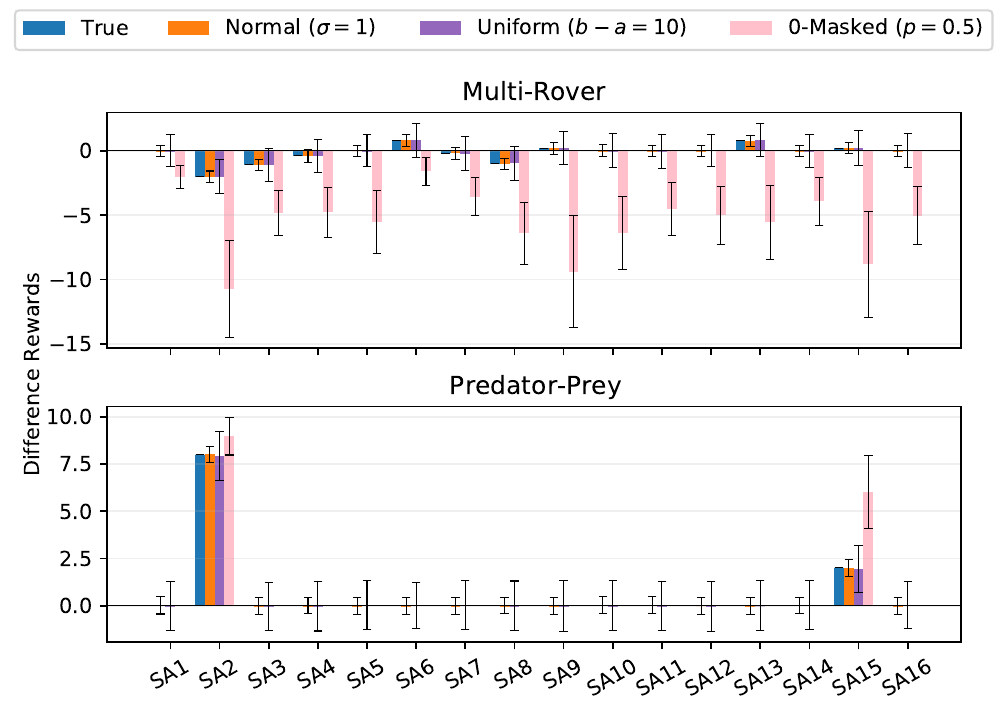}
\caption{Mean and variance of difference rewards for a set of samples under different noise profiles.}
\label{fig:noise}
\end{figure}

It can be observed how different noise processes differently affect the resulting difference rewards. For example, in both environments, the difference rewards mechanism is quite resistant against noise from a normal or a uniform distribution. This is probably due to the symmetricity of these noises, that tends to cancel out with each other. However, a masking kind of noise, under which some of the reward values are replaced with a value of $0$ with a certain probability, seems to be more detrimental for difference rewards evaluation: cancelling out some of the reward values definitely changes the computation and gives wrong estimates. This is worse in the multi-rover domain, in which the reward function is dense, while for the predator-prey environment and its sparse reward function it seems to be less harming.

These two observations together help explain why Dr.ReinforceR outperforms COMA on the two proposed environments: learning the reward function $R(s,a)$ is easier than learning the $Q$-function and, although function approximation introduces noise, the difference rewards mechanism is resistant against common types of noise and still provides useful signals to policy gradients. Therefore, if one is able to learn a good approximation of the reward, the proposed algorithm learns better and more reliable policies than other policy gradient algorithms, without the difficulties of learning the $Q$-function.

\section{Partial Observability}
Full observability of the environment as in MMDPs is a desirable property, but in many real-world situations \citep{sensors,traffic,smac} such a strong assumption is often unrealistic. The complexity of the environment itself or the limited sensing or communication capabilities available are usually transforming such problems into a partially observable ones from the perspective of the agents. In these, the agents cannot directly observe the state of the environment, but instead are provided with a local and possibly noisy observation that represents only a limited amount of information about the underlying environment state itself.
 
Formally, such settings can be modelled as a decentralized partially observable Markov decision process (Dec-POMDP) \citep{decpomdp} $\mathcal{M}=\langle D,S,\{A^i\}_{i=1}^{\vert D\vert },T,R,\{O^i\}_{i=1}^{\vert D\vert },Z,\gamma\rangle$, where $D,S,A^i,T,R$ and $\gamma$ are the same as in a MMDP. As mentioned above, agents are provided with a local observation $o^i\in O^i$, such that $o=\langle o^1,\ldots,o^N\rangle\in\times_{i=1}^{\vert D\vert }O^i=O$ is called a joint observation and $o\sim Z(s)$, where $Z:S\rightarrow O$ is the observation function. With such limitations, each agent has to keep track of its own action-observation history $h^i_t=(o^1_0,a^i_0,o^i_1,a^i_1,\ldots,o^i_{t-1},a^i_{t-1},o^i_t)\in (O^i\times A^i)^*\times O^i=\mathcal{H}^i$ up to the current time step $t$ to try and assess the underlying state of the environment, and use this to condition its policy and draw its decisions. A joint history at time step $t$ can also be defined as $h_t=(o_0,a_0,o_1,a_1,\ldots,o_{t-1},a_{t-1},o_t)\in (O\times A)^*\times O=\mathcal{H}$.

Policy gradients algorithms can easily be adapted to work under partial observability by simply replacing the environment state $s$ used by the agents policies $\pi_{\theta^i}$ with the corresponding agent's local action-observation history $h^i_t$. The distributed policy gradients in Equation \eqref{eq:mapg} thus becomes:
\begin{equation}
\theta^i\leftarrow\theta^i+\alpha\underbrace{\sum_{t=0}^{T-1}\gamma^tG_t\nabla_{\theta^i}\log\pi_{\theta^i}(a^i_t\vert h^i_t)}_{\hat{g}^i}.
\end{equation}

\subsection{Method}
Similarly, it is straightforward to also adapt Dr.Reinforce to work in Dec-POMDPs by simply adjusting the policy terms that appear in Equation \eqref{eq:diff_return} and Equation \eqref{eq:drpg} to condition on the agents' local action-observation histories $h^i_t$. The difference return $\Delta G^i_t$ in thus defined as:

\begin{equation}
\Delta G^i_t(a_{t:T}^i\vert s_{t:T},a_{t:T}^{-i},h^i_{t:T})\triangleq\sum_{l=0}^{T-t-1}\gamma^l\Delta R^i(a_{t+l}^i\vert s_{t+l},a_{t+l}^{-i},h^i_{t+l}),
\label{eq:diff_return_po}
\end{equation}
while the decentralized policies are learned by using the update target:
\begin{equation}
\theta^i\leftarrow\theta^i+\alpha\underbrace{\sum_{t=0}^{T-1}\gamma^t\Delta G^i_t(a_{t:T}^i\vert s_{t:T},a_{t:T}^{-i},h^i_{t:T})\nabla_{\theta^i}\log\pi_{\theta^i}(a^i_t\vert h^i_t)}_{g^{DR,i}}.
\label{eq:drpg_po}
\end{equation}

When complete access to the reward function is not available, a modified version of Dr.ReinforceR can be applied. The centralized reward network $R_{\psi}$, by following the CTDE paradigm, can still be learned in the same way as in Equation \eqref{eq:r} and condition on the environment state $s\in S$, as it is not required during execution. It is enough to adapt Equation \eqref{eq:ar_learned} as done before, thus obtaining:

\begin{equation}
\Delta R_{\psi}^i(a^i_t\vert s_t,a_t^{-i},h^i_t)\triangleq r_t-\sum_{c^i\in A^i}\pi_{\theta^i}(c^i\vert h^i_t)R_{\psi}(s_t,\langle c^i,a^{-i}_t\rangle),
\end{equation}
and consequently adjust Equation \eqref{eq:diff_return_learned} as:

\begin{equation}
\Delta G^i_t(a_{t:T}^i\vert s_{t:T},a_{t:T}^{-i},h^i_{t:T})\triangleq\sum_{l=0}^{T-t-1}\gamma^l\Delta R_{\psi}^i(a_{t+l}^i\vert s_{t+l},a_{t+l}^{-i},h^i_t).
\label{eq:dr_learned}
\end{equation}

\subsection{Theoretical Results}
Above, we adapted Dr.Reinforce, which intuitively can improve learning by providing individual agents with a better learning signal, to partially observable settings. In these, using difference rewards as the agents' learning signals induces a partially observable stochastic game \citep{posg,decpomdp} $\hat{\mathcal{P}}=\langle D,S,\{A^i\}_{i=1}^{\vert D\vert },T,$ $\{\Delta R^i\}_{i=1}^{\vert D\vert },\{O^i\}_{i=1}^{\vert D\vert },Z\rangle$ in which the cooperating agents do not receive the same reward after each time step. Even though difference rewards are aligned with the true reward values \citep{visualizing,invariance}, for these games convergence to an optimal solution is not immediate.

When agents are required to base their decisions on their local action-observation history $h^i_t$, the same result on an unbiased baseline derived in Section \ref{sec:theory} for the fully observable case does not hold anymore. Generally speaking, this is due to the Monte-Carlo nature of the difference return $\Delta G^i_t$, that requires future quantities in order to compute the value of the baseline. The local histories for the episode time steps (used to compute the aristocrat utility values in the r.h.s. of Equation \eqref{eq:diff_return_po}) are now strictly depending on the actions selected at the previous time steps, and thus break this independence of the baseline from the current action selection.

\begin{observation}
In a Dec-POMDP setting, using difference return $\Delta G^i_t(a_{t:T}^i\vert s_{t:T},a_{t:T}^{-i},h^i_{t:T})$ as the learning signal for policy gradients in Equation \eqref{eq:drpg_po} is in general not equivalent to subtracting an unbiased baseline $B^i(s_{t:T},a^{-i}_{t:T},h^i_{t:T})$ from the distributed policy gradients in Equation \eqref{eq:mapg}.
\label{thm:wrong}
\end{observation}

\begin{proof}
We start by rewriting $\Delta G^i_t(a_{t:T}^i\vert s_{t:T},a_{t:T}^{-i},h^i_{t:T})$ from Equation \eqref{eq:diff_return_po} as:

\begin{align}
\Delta G^i_t(a_{t:T}^i\vert s_{t:T},a_{t:T}^{-i},h^i_{t:T}) & =\sum_{l=0}^{T-t-1}\gamma^lr_{t+l} \nonumber\\
& -\sum_{l=0}^{T-t-1}\gamma^l\sum_{c^i\in A^i}\pi_{\theta^i}(c^i\vert h^i_{t+l})R(s_{t+l},\langle c^i,a^{-i}_{t+l}\rangle).
\label{eq:diff_return_po_div}
\end{align}

Note that the first term on the r.h.s. of Equation \eqref{eq:diff_return_po_div} is the return $G_t$ used in Equation \eqref{eq:mapg}. We then define the second term on the r.h.s. of Equation \eqref{eq:diff_return_po_div} as the baseline $B^i(s_{t:T},a^{-i}_{t:T},h^i_{t:T})$:

\begin{equation}
B^i(s_{t:T},a^{-i}_{t:T},h^i_{t:T})=\sum_{l=0}^{T-t-1}\gamma^l\sum_{c^i\in A^i}\pi_{\theta^i}(c^i\vert h^i_{t+l})\cdot R(s_{t+l},\langle c^i,a^{-i}_{t+l}\rangle).
\label{eq:baseline_po}
\end{equation}

We can thus rewrite the total expected update target for agent $i$ as:

\begin{align}
\mathbb{E}_{\pi_{\theta}}\left[\hat{g}^{DR,i}\right] & =\mathbb{E}_{\pi_{\theta}}\left[\sum_{t=0}^{T-1}\left(\nabla_{\theta^i}\log\pi_{\theta^i}(a^i_t\vert h^i_t)\right)\Delta G^i_t(a_{t:T}^i\vert s_{t:T},a_{t:T}^{-i},h^i_{t:T})\right] \nonumber\\
& =\mathbb{E}_{\pi_{\theta}}\left[\sum_{t=0}^{T-1}\left(\nabla_{\theta^i}\log\pi_{\theta^i}(a^i_t\vert h^i_t)\right)\left(G_t-B^i(s_{t:T},a^{-i}_{t:T},h^i_{t:T})\right)\right] \nonumber\\
& \text{(by definition of }\Delta G^i_t) \nonumber\\
& =\mathbb{E}_{\pi_{\theta}}\left[\sum_{t=0}^{T-1}\left(\nabla_{\theta^i}\log\pi_{\theta^i}(a^i_t\vert h^i_t)\right)G_t\right. \nonumber\\
& \left.-\left(\nabla_{\theta^i}\log\pi_{\theta^i}(a^i_t\vert h^i_t)\right)B^i(s_{t:T},a^{-i}_{t:T},h^i_{t:T})\right] \nonumber\\
& \text{(distributing the product)} \nonumber\\
& =\mathbb{E}_{\pi_{\theta}}\left[\sum_{t=0}^{T-1}\left(\nabla_{\theta^i}\log\pi_{\theta^i}(a^i_t\vert h^i_t)\right)G_t\right] \nonumber\\
& -\mathbb{E}_{\pi_{\theta}}\left[\sum_{t=0}^{T-1}\left(\nabla_{\theta^i}\log\pi_{\theta^i}(a^i_t\vert h^i_t)\right)B^i(s_{t:T},a^{-i}_{t:T},h^i_{t:T})\right] \nonumber\\
& \text{(by linearity of the expectation)} \nonumber\\
& =\mathbb{E}_{\pi_{\theta}}\left[\hat{g}^{PG,i}\right]+\mathbb{E}_{\pi_{\theta}}\left[\hat{g}^{B,i}\right]. \nonumber
\end{align}

In order to show that the baseline is unbiased the expected value of its update $\mathbb{E}_{\pi_{\theta}}\left[\hat{g}^{B,i}\right]$ with respect to the policy $\pi_{\theta}$ should be $0$. Let $P^{\pi_{\theta}}(h_t)=P^{\pi_{\theta}}(h_{t-1})\cdot\pi_{\theta}(a_{t-1}\vert h_{t-1})\sum_{s_t\in S}P^{\pi_{\theta}}_t(s_t)\cdot Z(o_t,s_t)$ (with $P^{\pi_{\theta}}(h_0)=\sum_{s_0\in S}Z(o_0\vert s_0)\rho(s_0)$ and $\rho(s_0)$ the initial state distribution) be the joint action-observation history distribution. Let also define the \emph{complete system history} $\hat{h}_t=\langle h_t,a_t,s_{0:t}\rangle\in\hat{\mathcal{H}}_t$, so that $P^{\pi_{\theta}}(\hat{h}_t)=P^{\pi_{\theta}}(h_t)\cdot \pi_{\theta}(a_t\vert h_t)\cdot\prod_{l=0}^tP^{\pi_{\theta}}_l(s_l)$, we have:

\begin{align}
\mathbb{E}_{\pi_{\theta}}\left[\hat{g}^{B,i}\right] & \;\triangleq-\mathbb{E}_{\pi_{\theta}}\left[\sum_{t=0}^{T-1}\left(\nabla_{\theta^i}\log\pi_{\theta^i}(a^i_t\vert h^i_t)\right)B^i(s_{t:T},a^{-i}_{t:T},h^i_{t:T})\right] \nonumber\\
& \;=-\sum_{t=0}^{T-1}\sum_{\hat{h}_t\in\hat{\mathcal{H}}_t}P^{\pi_{\theta}}(\hat{h}_t)\left(\nabla_{\theta^i}\log\pi_{\theta^i}(a^i_t\vert h^i_t)\right) \nonumber\\
& \sum_{\hat{h}_T\in\hat{\mathcal{H}}_T}P^{\pi_{\theta}}(\hat{h}_T\vert\hat{h}_t)\; B^i(s_{t:T},a^{-i}_{t:T},h^i_{t:T}) \nonumber\\
& \text{(by expanding the expectation)} \nonumber\\
& \;=-\sum_{t=0}^{T-1}\sum_{h_t\in\mathcal{H}_t}P^{\pi_{\theta}}(h_t)\sum_{a^{-i}_t\in A^{-i}}\pi_{\theta^{-i}}(a^{-i}_t\vert h^{-i}_t) \nonumber\\
& \sum_{a^i_t\in A^i}\left(\nabla_{\theta^i}\pi_{\theta^i}(a^i_t\vert h^i_t)\right)\sum_{\hat{h}_T\in\hat{\mathcal{H}}_T}P^{\pi_{\theta}}(\hat{h}_T\vert\hat{h}_t)\; B^i(s_{t:T},a^{-i}_{t:T},h^i_{t:T}) \nonumber\\
& \text{(by applying the inverse log trick)} \nonumber\\
& \;\neq-\sum_{t=0}^{T-1}\sum_{h_t\in\mathcal{H}_t}P^{\pi_{\theta}}(h_t)\sum_{a^{-i}_t\in A^{-i}}\pi_{\theta^{-i}}(a^{-i}_t\vert h^{-i}_t) \nonumber\\
& \left(\nabla_{\theta^i}\sum_{a^i_t\in A^i}\pi_{\theta^i}(a^i_t\vert h^i_t)\right)\sum_{\hat{h}_T\in\hat{\mathcal{H}}_T}P^{\pi_{\theta}}(\hat{h}_T\vert\hat{h}_t)\; B^i(s_{t:T},a^{-i}_{t:T},h^i_{t:T})) \nonumber\\
& \text{(by moving the gradient outside the policy sum)} \nonumber
\end{align}

We cannot move the gradient outside of the sum now (as done in Equation \eqref{eq:derivation}), because of the baseline $B^i$ depending on the policy parameters via the agent action $a^i_t$ included in the histories $h^i_{t+1:T}$. The sum over the policy term is therefore a weighted summation over different baseline values, and these in general do not sum up to $0$, and thus the baseline is in general not unbiased (although problems for which the summation is $0$ in any case may exist, and in these special cases the baseline is still unbiased).
\end{proof}

The result in the above Lemma shows that using the baseline in Equation \eqref{eq:baseline_po} alter the expected value of the overall gradient, as the baseline $B^i(s_{t:T},a^{-i}_{t:T},h^i_{t:T})$ is not unbiased, and thus the policy gradients are not guaranteed to converge to the same solutions of the distributed policy gradients \citep{mapg}.

\subsection{StarCraftII Experiments}
\begin{figure}[htbp]
\vspace{-2\baselineskip}
\centering
\subfloat{
\includegraphics[width=0.42\textwidth]{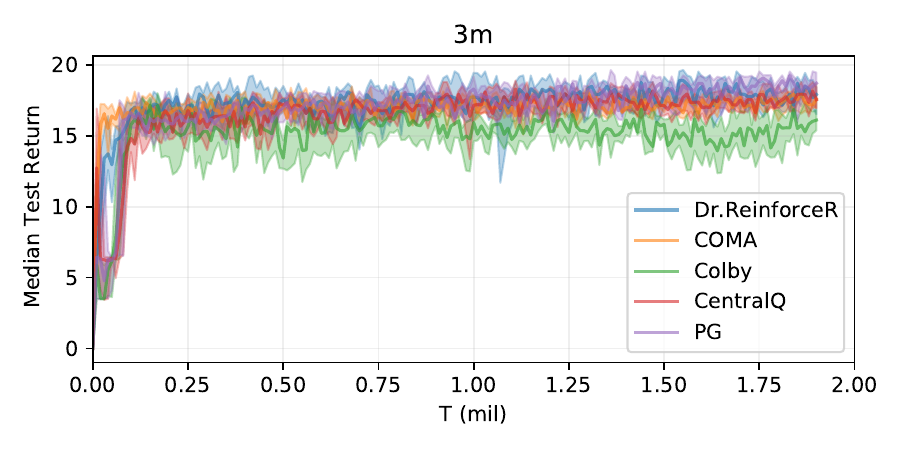}
}
\subfloat{
\includegraphics[width=0.42\textwidth]{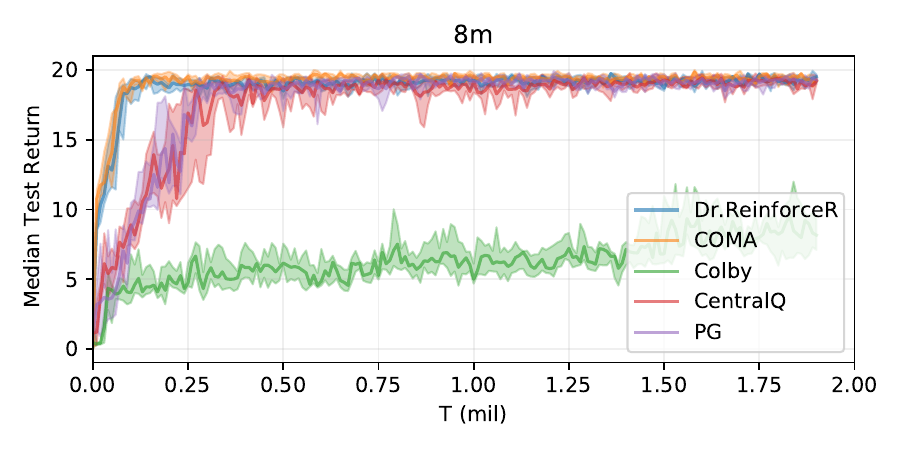}
}

\subfloat{
\includegraphics[width=0.42\textwidth]{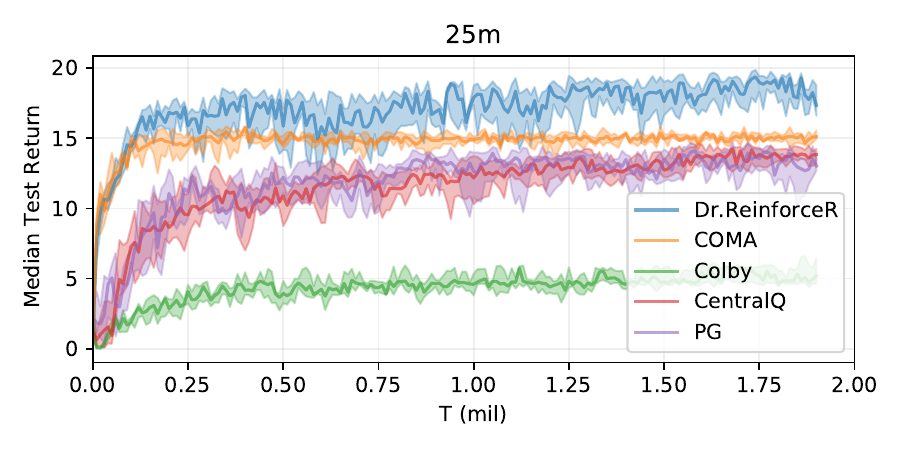}
}
\subfloat{
\includegraphics[width=0.42\textwidth]{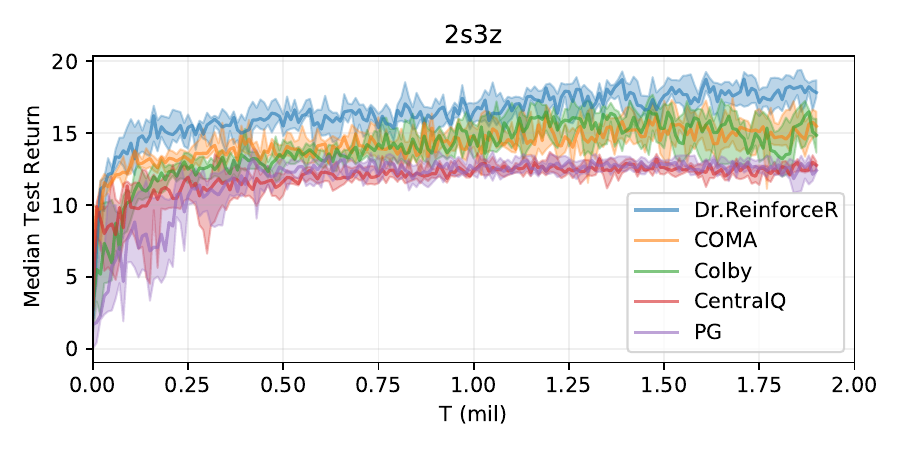}
}

\subfloat{
\includegraphics[width=0.42\textwidth]{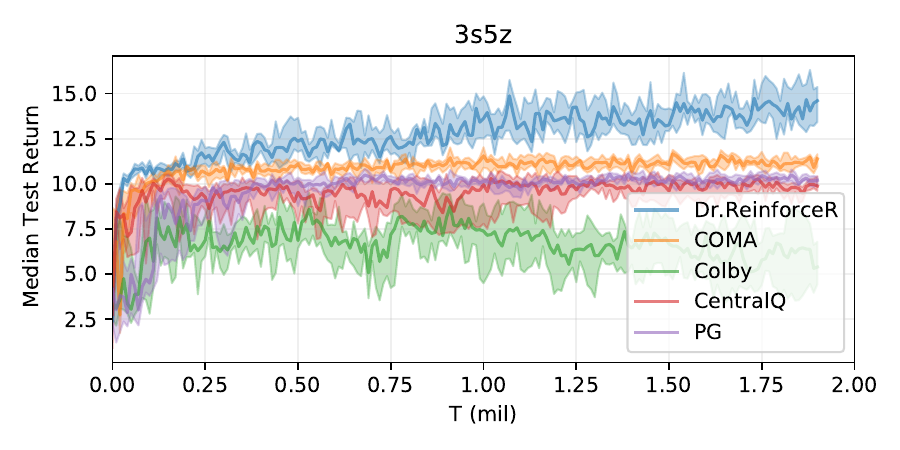}
}
\subfloat{
\includegraphics[width=0.42\textwidth]{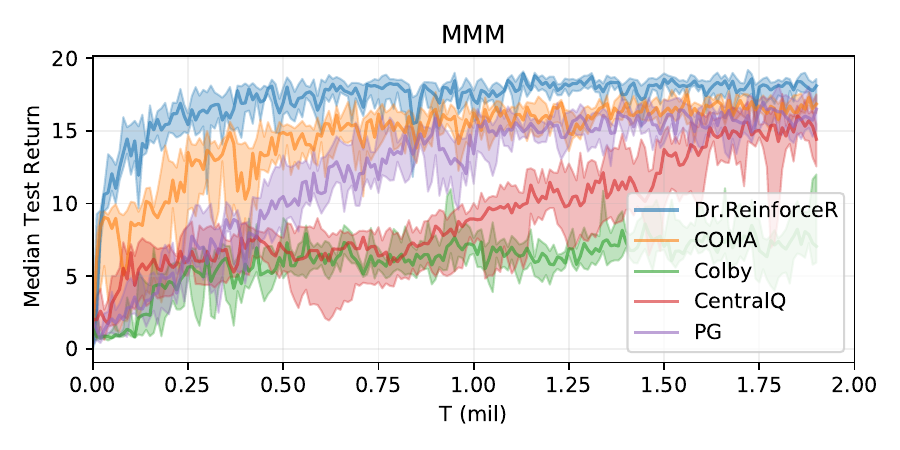}
}

\subfloat{
\includegraphics[width=0.42\textwidth]{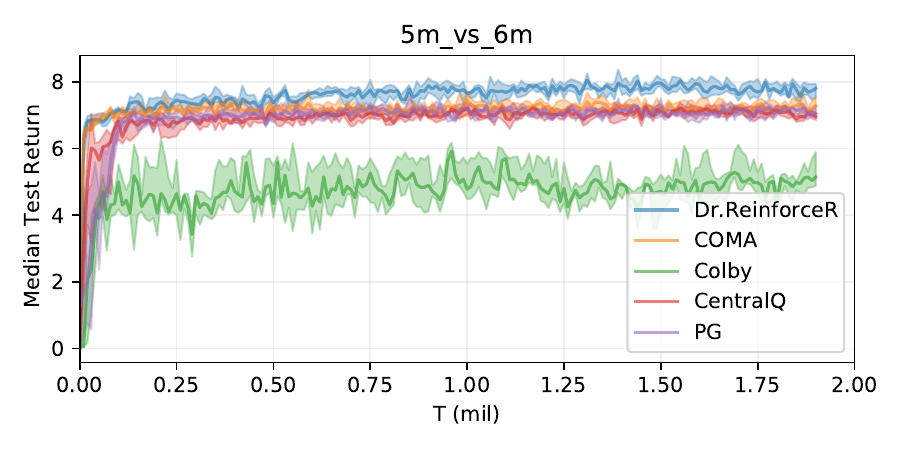}
}
\subfloat{
\includegraphics[width=0.42\textwidth]{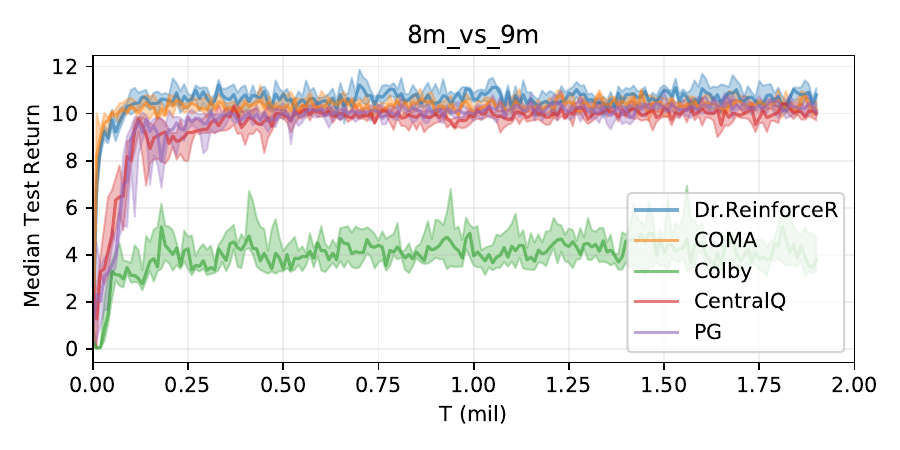}
}

\subfloat{
\includegraphics[width=0.42\textwidth]{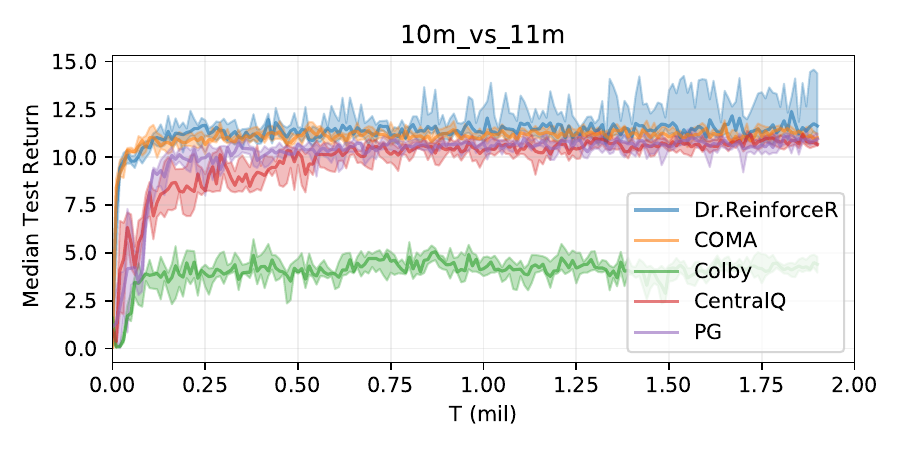}
}
\subfloat{
\includegraphics[width=0.42\textwidth]{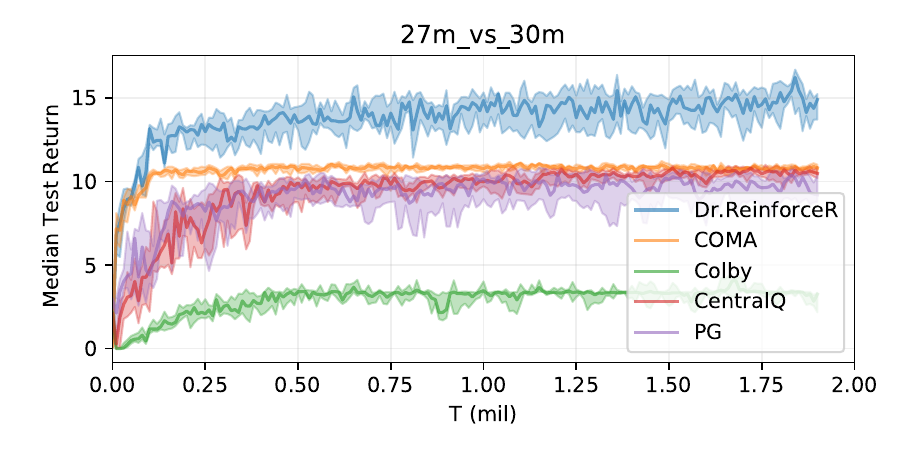}
}

\subfloat{
\includegraphics[width=0.42\textwidth]{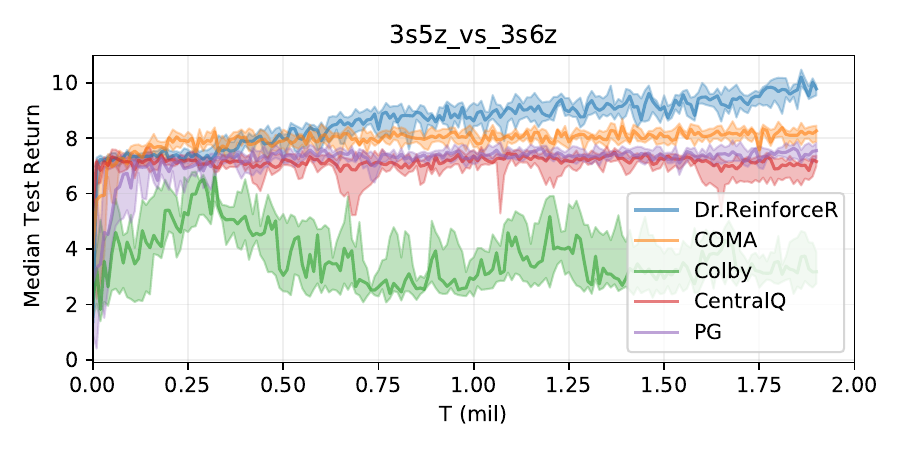}
}
\subfloat{
\includegraphics[width=0.42\textwidth]{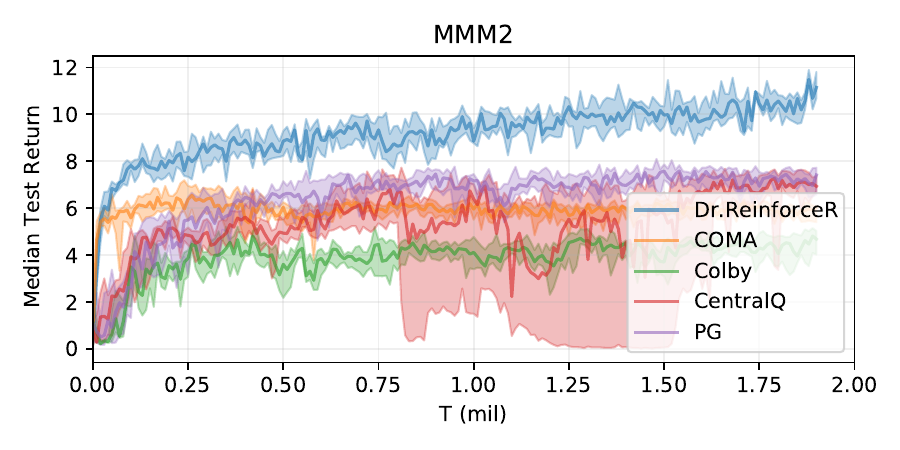}
}
\end{figure}

\begin{figure}[htbp]
\vspace{-2\baselineskip}
\centering
\subfloat{
\includegraphics[width=0.42\textwidth]{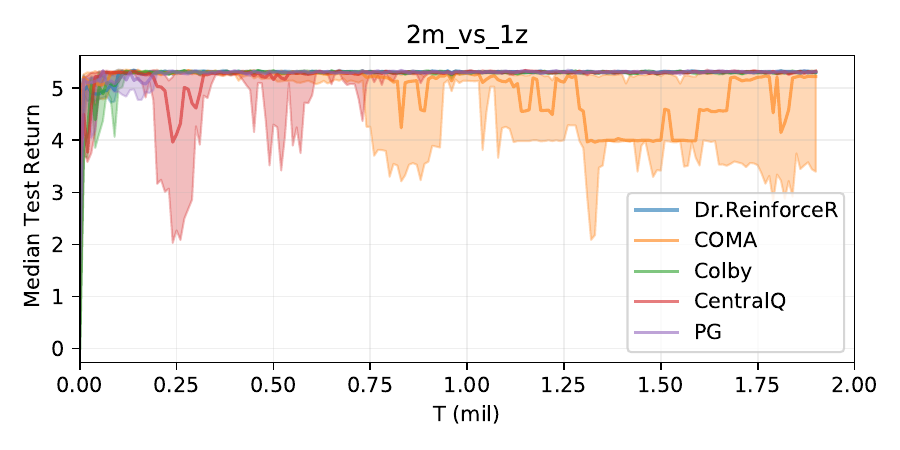}
}
\subfloat{
\includegraphics[width=0.42\textwidth]{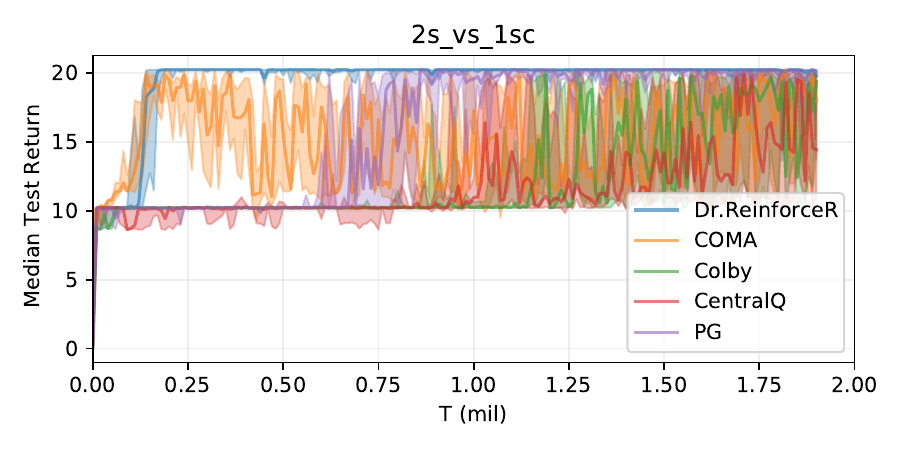}
}

\subfloat{
\includegraphics[width=0.42\textwidth]{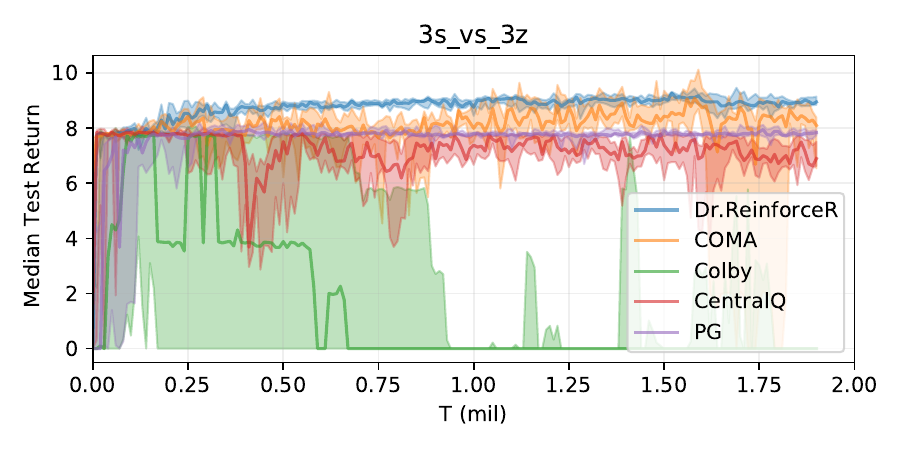}
}
\subfloat{
\includegraphics[width=0.42\textwidth]{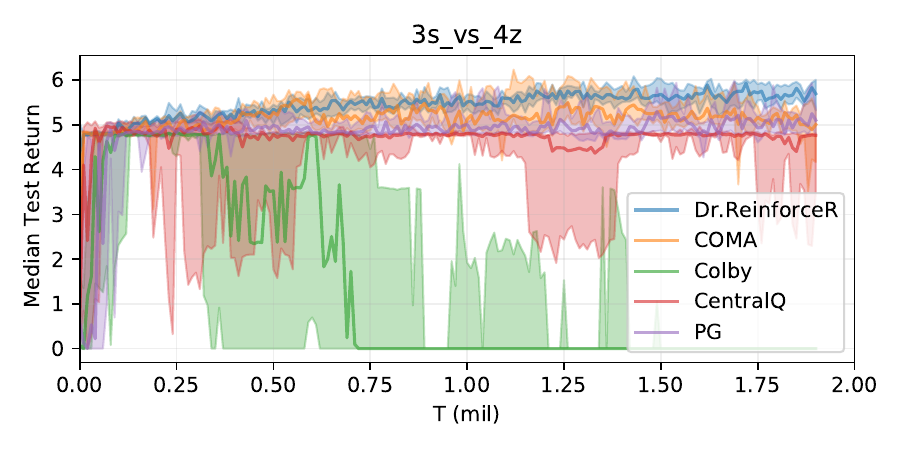}
}

\subfloat{
\includegraphics[width=0.42\textwidth]{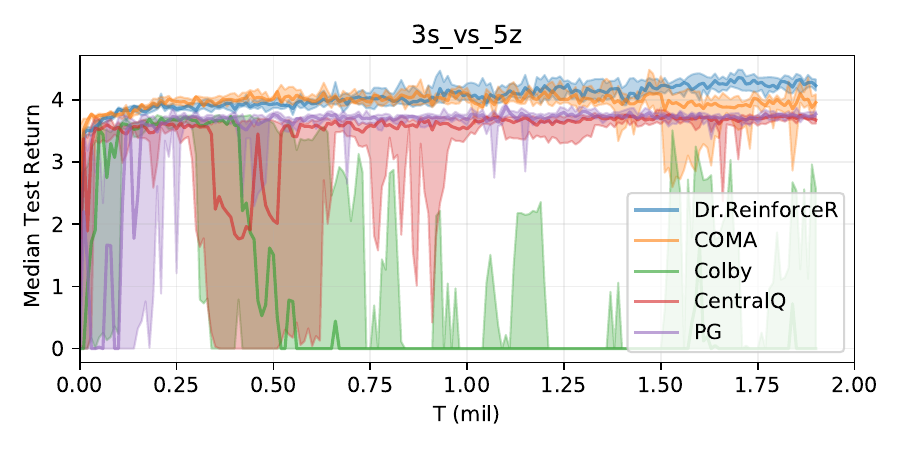}
}
\subfloat{
\includegraphics[width=0.42\textwidth]{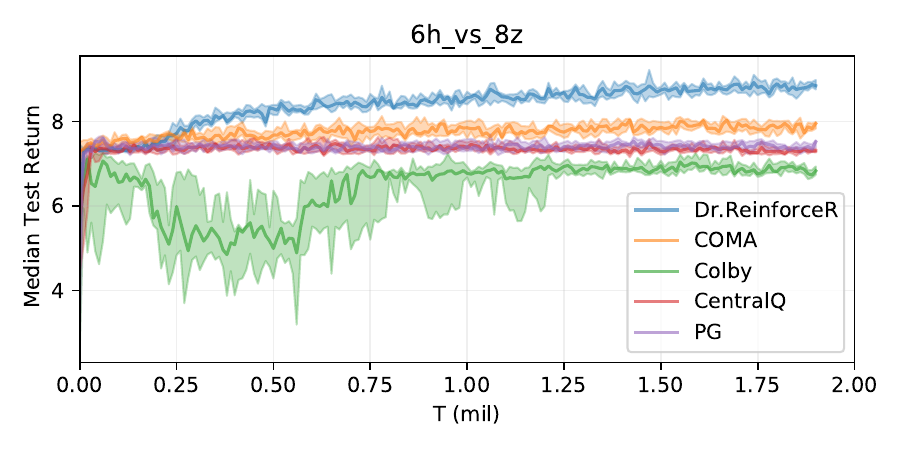}
}

\subfloat{
\includegraphics[width=0.42\textwidth]{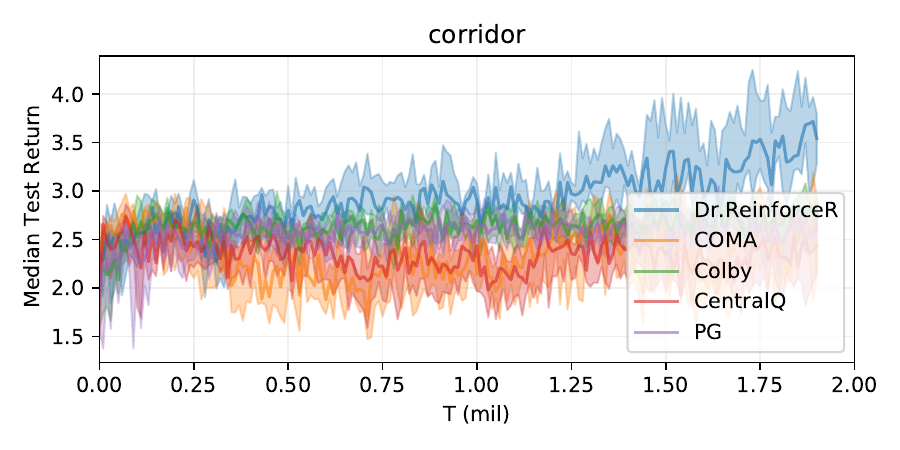}
}
\subfloat{
\includegraphics[width=0.42\textwidth]{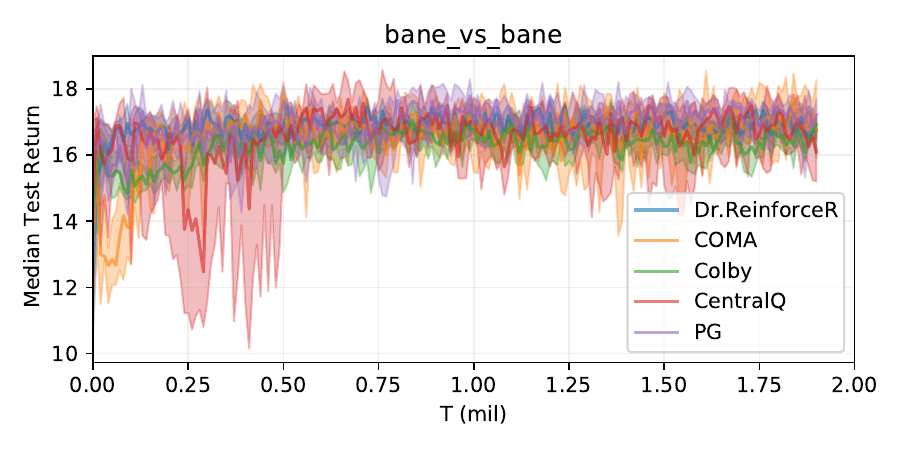}
}

\subfloat{
\includegraphics[width=0.42\textwidth]{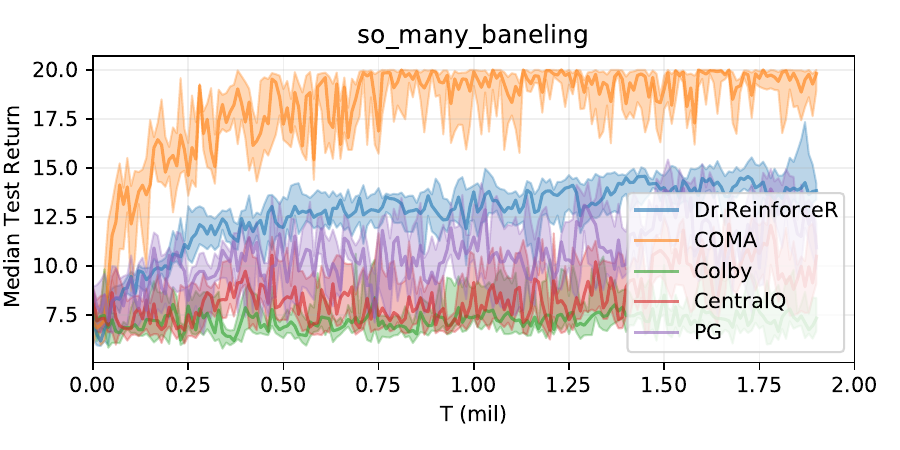}
}
\subfloat{
\includegraphics[width=0.42\textwidth]{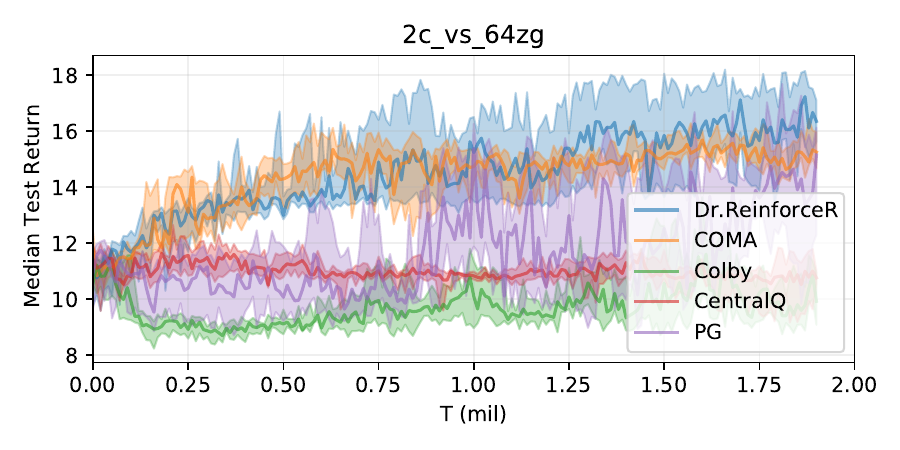}
}

\subfloat{
\includegraphics[width=0.42\textwidth]{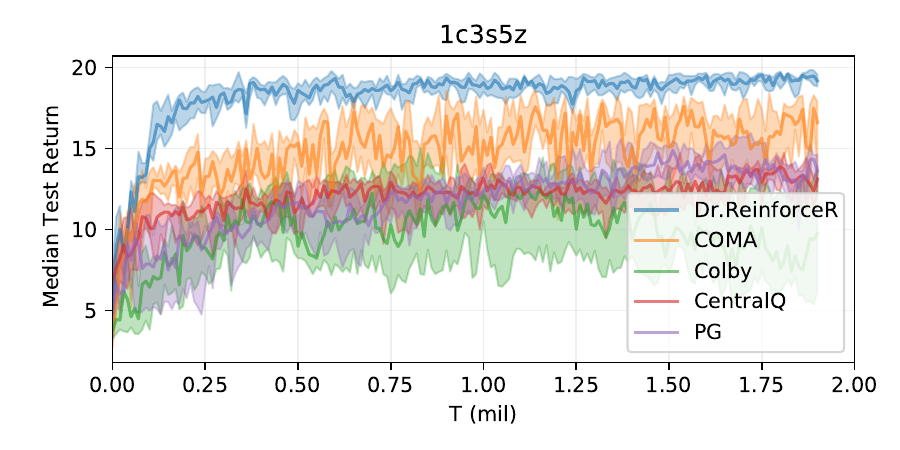}
}
\caption{Training curves on the entire set of SMAC maps, showing the median return and $25-75\%$ percentiles across seeds.}
\label{fig:smac}
\end{figure}

Although there is no theoretical guarantee on the convergence of our proposed method under partial observability, it might still work well in practice. Therefore, we investigate the application of our method on the StarCraftII multi-agent challenge (SMAC) \citep{smac}, a very complex, partially observable environment that provides a wide set of different maps, each with a different number and different types of units that has to fight against an opposing team controlled by the game AI, to show good empirical performances. As with the current game back-end \citep{sc2} it has not been possible to obtain all the reward values for the possible agents actions, we have not been able to apply Dr.Reinforce here. Figures \ref{fig:smac} shows median return and $25-75\%$ percentiles across $10$ independent runs on the whole set of available maps, with the difficulty level of the opponent team set to Very Hard.

In this setting, Dr.ReinforceR is almost never underperforming with respect to all the other baselines, with significant improvements over COMA on heterogeneous maps like \texttt{3s5z}, \texttt{1c3s5z} or \texttt{MMM}. This shows how learning the $Q$-function may be difficult in complex setting, while the reward network is easier to learn and in turn produces better policies. Also, it is worth mentioning that the severe partial observability of this setting is well addressed in practice by our use of the CTDE paradigm, with the reward network conditioned on the true state $s$: these results show how advantageous it is to resort to centralized training of the reward network over a local approximation as in the algorithm from \citep{reward2}. In particular, the good performance on the \texttt{25m} map, involving a large number of agents, shows again the better scalability of the proposed centralized reward network with respect to a centralized $Q$-function critic, where the effects of bootstrapping and the moving target problem become even more severe when the number of agents grows larger.

A noticeable exception is represented by the \texttt{so\textunderscore many\textunderscore baneling} map, where COMA is achieving good results, where neither Dr.ReinforceR and all the other baselines are outperformed. An hypothesis for this is that the difference return $\Delta G^i_t$ is driving each agent into performing the more rewarding actions at each step (for example, hit an opponent if possible), but in the long run this strategy is not a winning one on this particular map, and thus the agents never experience the high reward for winning and are thus never able to change their learned behaviours. Reasoning on the more complex $Q$-function here could be helpful to drive the policies towards a winning situation at the cost of performing actions that seem suboptimal at the current step. In Appendix \ref{sec:smac} we also report the median win rate obtained by the investigated algorithms. From these, we can observe that, even when Dr.ReinforceR is capable of learning high return policies, these may not be sufficient to also achieve a significant win rate in some scenarios (for\FloatBarrier example, on more challenging maps with asymmetrical teams, like \texttt{6h\textunderscore vs\textunderscore 6h} or \texttt{MMM2}, although the gap in achieved median returns with respect to all the other baselines is very significant).

\section{Related Work}
Application of reinforcement learning techniques to multi-agent systems has a long and fruitful history \citep{marl}. Fundamental works like \citep{qlearn} were the first to investigate the applicability of these algorithms in the form of independent learners to cooperative settings, while \citep{dynamics} further analyses the dynamics of their learning process depending on their consideration of the others. Specific algorithms to improve performance by learning the value of cooperation and coordination has been proposed, like in \citep{coordinated}. Also policy gradients has been widely applied to cooperative settings: \citep{mapg} first proved convergence of distributed policy gradients to the same solution obtained by a centralized agent. Closer to our approach are recent works of policy gradients with deep reinforcement learning: for example, \citep{coma} presents COMA, that efficiently estimates a counterfactual baseline for a team of cooperating homogeneous agents using a centralized critic for discrete problems. \citep{po} takes inspiration from game theory and regret minimization to design a family of algorithms based on counterfactual regret minimization for partially observable domains. \citep{consensus} combines actor-critic with a consensus mechanism to solve cooperative problems when communication is available, and also provide convergence proof under certain conditions, while \citep{dop} combines value-decomposition with a counterfactual baseline in the actor-critic framework. All the above algorithms use the action-value function in order to compute the counterfactuals, that can be difficult to learn because of bootstrapping target problems. Our method on the other hand learns the reward function to approximate the difference rewards, that do not suffer from these problems. For a more extensive review on recent deep reinforcement learning algorithms for cooperative multi-agent systems see \citep{ctde} and \citep{critique}.

Another important line of work for us is that on difference rewards \citep{aristocrat}, that already served as a basis for some existing algorithms like COMA. \citep{airflow} uses difference rewards in learning to control a fleet of air vehicles that has to coordinate on traffic routes. \citep{credit} proposes two difference rewards-based value-functions to improve multi-agent actor-critic in the $\mathbb{C}$Dec-POMDP setting, while \citep{potential} combines difference rewards and dynamic potential-based reward shaping \citep{theoretical,shaping} to improve performance and convergence speed. Also, \citep{objective} applies difference rewards to multi-objective problems, speeding up learning and improving performance. Finally, some works try to improve the standard definition of difference rewards: \citep{difference} proposes to approximate difference rewards using tabular linear functions when it is not possible to access the value of the reward for the default action through a simulator, while \citep{reward1} and \citep{reward2} both propose to approximate the difference rewards by using only local information. With the exception of the latter, the aforementioned works all uses value-based algorithms to learn, while our method resorts to a policy gradients algorithm, that recently showed great premise in multi-agent learning contexts.

Finally, the idea of learning the reward function has also received some attention, especially in the single-agent setting. \citep{reward} learns an additional state-reward network to reduce variance when updating the value-function in noisy environments, \citep{local} uses Kalman filters in problems with noise coming from different sources to explicitly learn about the reward function and the noise term, while \citep{auxiliary} proposes UNREAL, that additionally learn to predict rewards as an auxiliary task to improve deep reinforcement learning agent performance. Finally, \citep{factored} learns a factored reward representation for multi-agent cooperative one-shot games. While these works learn the reward function, these are mainly limited to the single-agent setting (with the exceptions of \citep{local} and \citep{factored}, which analyse different aspects from our and can be considered orthogonal and used in conjunction with our work) and do not use it to approximate the difference rewards.

\section{Discussion and Future Work}
Despite the good empirical results obtained by Dr.ReinforceR in the experiments detailed above, Lemma \ref{thm:wrong} clearly shows that the combination of difference rewards and policy gradients in a partially observable setting has in general no theoretical guarantees of convergence, as the baseline that is subtracted from the distributed policy gradients is not unbiased. This means that experimental performance could be unstable or arbitrarily bad.

Here we try and identify possible alternatives to our investigated formulation that are capable of restoring the theoretical convergence guarantees. This could be ensured by replacing the current baseline $B^i(s_{t:T},a^{-i}_{t:T},h^i_{t:T})$ in Equation \eqref{eq:baseline_po} with a new $\tilde{B}^i(s_{t:T},a^{-i}_{t:T},h^i_t)$ that does not depend on the currently selected action $a^i_t$ via the local histories $h^i_{t+1:T}$. We identified a couple of possible solutions, that are not however investigated in the current paper:

\begin{enumerate}
\item Replace the current agent policy $\pi_{\theta^i}(a^i_t\vert h^i_t)$ with a fixed policy $\mu(a^i_t)$ (a type of difference rewards also proposed in \citep{aristocrat}):
$$\tilde{B}^i(s_{t:T},a^{-i}_{t:T})=\sum_{l=0}^{T-t-1}\gamma^l\sum_{c^i\in A^i}\mu(c^i)\cdot R(s_{t+l},\langle c^i,a^{-i}_{t+l}\rangle).$$
This idea however would require to fix beforehand a policy $\mu(a^i_t)$ to use, a choice similar to that of the default action \citep{aristocrat,difference} in Equation \eqref{eq:dr}.
\item Use the current agent policy $\pi_{\theta^i}(a^i_t\vert h^i_t)$, but do not condition on the local histories for the episode time steps $h^i_{t+1:T}$, but only on the current local history $h^i_t$:
$$\tilde{B}^i(s_{t:T},a^{-i}_{t:T},h^i_t)=\sum_{l=0}^{T-t-1}\gamma^l\sum_{c^i\in A^i}\pi_{\theta^i}(c^i\vert h^i_t)\cdot R(s_{t+l},\langle c^i,a^{-i}_{t+l}\rangle).$$
\item Use a potential-based reward shaping mechanism. These are known to retain policy invariance in single-agent reinforcement learning, both under full observability \citep{invariance} as well as partial one \citep{pomdp}, while in multi-agent systems converge to the same set of Nash Equilibria of the policies learned with the shared reward alone \citep{theoretical,shaping}, while improve learning performance. In general, a potential-based reward shaping mechanism provides the agents with a shaped reward $\hat{r}$:

\begin{equation*}
\hat{r}\triangleq r_t+\underbrace{F(s_t,s_{t+1})}_{\tilde{B}^i},
\end{equation*}
where $F(s_t,s_{t+1})=\gamma\phi(s_{t+1})-\phi(s_t)$, and $\phi(s)$ is a suitable function that provides additional information on the state $s$, so that $F(s_t,s_{t+1})$ is unbiased in expectation with respect to the policy gradients, and thus keep the convergence guarantees.

A particular form of potential-based reward shaping, that combines its benefit with those of difference rewards, is Counterfactual as Potential \citep{potential}, in which the potential-based reward shaping function is:

\begin{equation*}
\phi(s)=R(s^{-i}),
\end{equation*}
and $R(s^{-i})$ is a reward term that marginalizes out the presence of agent $i$. It is to note that, while in general such term needs to be provided by the environment itself via the use of a simulator (as with difference rewards), with our learned reward network that issue could be overcome.
\end{enumerate}

Another crucial aspect of Dr.ReinforceR is that it resorts to the CTDE framework \citep{planning,ctde} to learn its centralized reward network. Although CTDE is a widely used and accepted methodology \citep{coma,maddpg}, it indeed restricts the training procedure to be carried out offline and in a separate step from the agents execution. There are settings however in which being able to retain decentralized execution while being able to learn during real interactions with the environment may be required. In such cases, it may be appropriate to replace the centralized reward network $R_{\psi}$ with a set of individual reward networks $R_{\psi^i}(s,a^i)$ (or $R_{\psi^i}(h^i_t,a^i)$ when learning in a Dec-POMDP), one for each agent $i$, to approximate the difference rewards computation. These local networks are learning the expected value of the reward for each agent when performing a certain action in a given situation, independently of what the others are doing

\begin{equation*}
R_{\psi^i}(s,a^i)\approx\mathbb{E}_{\pi_{\theta^{-i}}}\left[R_{\psi}(s,\langle a^i,a^{-i}\rangle)\right].
\end{equation*}

This additional approximation is suitable to break the dependence from the CTDE paradigm, although it may introduce approximation error in the local reward terms via the expectation over the other agents policies (while the centralized reward network $R_{\psi}$ is in principle capable of perfectly approximate the reward function $R(s,a)$ and thus provide the policy gradients with perfect difference rewards values).

\section{Conclusions}
In cooperative multi-agent systems agents face the problem of figuring out how they are contributing to the overall performance of the team in which only a shared reward signal is available. Previous methods like COMA, a state-of-the-art difference rewards algorithm, used the action-value function to compute an individual signal for each agent to drive policy gradients. However, learning a centralized $Q$-function is problematic due to inherent factors like bootstrapping or the dependence on the joint action.

We proposed Dr.Reinforce, a novel algorithm that tackles multi-agent credit assignment by combining policy gradients and differencing of the reward function. When the true reward function is known, our method outperforms all compared baselines on two benchmark multi-agent cooperative environments with a shared reward signal, and scales much better with the number of agents, a crucial capability for real cooperative multi-agent scenarios.

Additionally, for settings in which such reward function is not known, we additionally proposed Dr.ReinforceR, that learns a centralized reward network used for estimating the difference rewards. Although the reward function has got the same dimensionality of the $Q$-function used by COMA, its learning is easier as no bootstrapping or moving target is involved. Although learning a reward network capable of appropriately generalizing across the state-action space may be challenging and have pitfalls, we showed how Dr.ReinforceR is able to outperform COMA, a state-of-the-art difference rewards algorithm, and achieve higher performance.

Therefore, exploring how to improve the representational capabilities of the reward network to allow it to better generalize to unseen situations and to be applicable to more complex scenarios is an interesting future direction that could further push the performance of these methods.

\section*{Acknowledgements}
\begin{wrapfigure}{r}{0.2\columnwidth}
\vspace{-17pt}
\includegraphics[width=0.2\columnwidth]{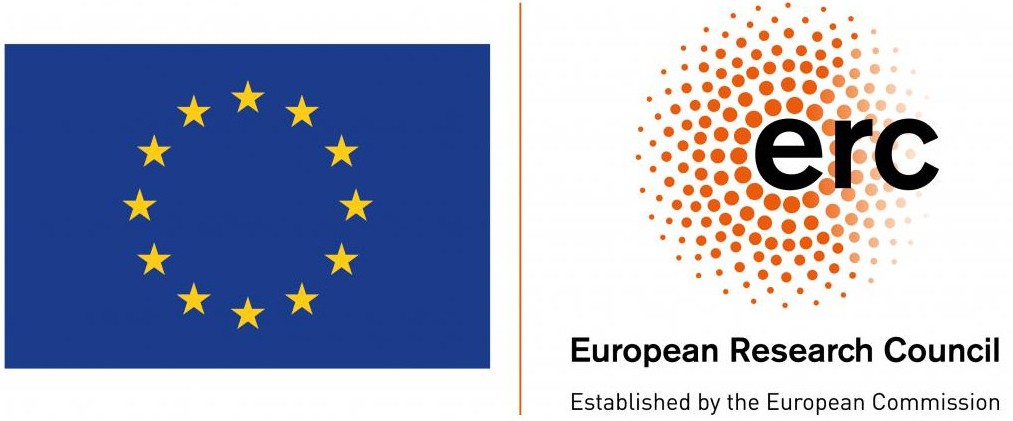}
\end{wrapfigure}
This work was supported by an Azure for Research computing grant. F.A.O.\ is funded by EPSRC First Grant EP/R001227/1.
This project received funding from the European Research Council (ERC) under the European Union's Horizon 2020 research and innovation programme (grant agreement No.~758824 \textemdash INFLUENCE).

\bibliographystyle{plainnat}
\bibliography{Bibliography}

\begin{thebibliography}{57}
\providecommand{\natexlab}[1]{#1}
\providecommand{\url}[1]{\texttt{#1}}
\expandafter\ifx\csname urlstyle\endcsname\relax
  \providecommand{\doi}[1]{doi: #1}\else
  \providecommand{\doi}{doi: \begingroup \urlstyle{rm}\Url}\fi

\bibitem[Agogino and Tumer(2008)]{visualizing}
Adrian~K. Agogino and Kagan Tumer.
\newblock Analyzing and visualizing multiagent rewards in dynamic and
  stochastic domains.
\newblock \emph{Autonomous Agents and Multi-Agent Systems}, 17:\penalty0
  320--338, 2008.

\bibitem[Bonferroni(1936)]{bonferroni}
Carlo~E. Bonferroni.
\newblock Teoria statistica delle classi e calcolo delle probabilit{\`a}.
\newblock \emph{Pubblicazioni del R Istituto Superiore di Scienze Economiche e
  Commerciali di Firenze}, 8:\penalty0 3--62, 1936.

\bibitem[Bottou(1998)]{sgd}
Léon Bottou.
\newblock Online learning and stochastic approximations, 1998.

\bibitem[Boutilier(1996)]{mmdp}
Craig Boutilier.
\newblock Planning, learning and coordination in multiagent decision processes.
\newblock In \emph{Proceedings of the 6th Conference on Theoretical Aspects of
  Rationality and Knowledge}, TARK ’96, pages 195–--210. Morgan Kaufmann
  Publishers Inc., 1996.

\bibitem[Busoniu et~al.(2008)Busoniu, Babuska, and De~Schutter]{marl}
Lucian Busoniu, Robert Babuska, and Bart De~Schutter.
\newblock A comprehensive survey of multiagent reinforcement learning.
\newblock \emph{IEEE Transactions on Systems, Man, and Cybernetics, Part C
  (Applications and Reviews)}, 38:\penalty0 156--172, 2008.

\bibitem[Cao et~al.(2013)Cao, Yu, Ren, and Chen]{mas}
Yongcan Cao, Wenwu Yu, Wei Ren, and Guanrong Chen.
\newblock An overview of recent progress in the study of distributed
  multi-agent coordination.
\newblock \emph{IEEE Transactions on Industrial Informatics}, 9\penalty0
  (1):\penalty0 427--438, 2013.

\bibitem[Castellini et~al.(2019)Castellini, Oliehoek, Savani, and
  Whiteson]{factored}
Jacopo Castellini, Frans~A. Oliehoek, Rahul Savani, and Shimon Whiteson.
\newblock The representational capacity of action-value networks for
  multi-agent reinforcement learning.
\newblock In \emph{Proceedings of the 18th International Conference on
  Autonomous Agents and Multiagent Systems}, AAMAS'19, pages 1862--1864.
  International Foundation for Autonomous Agents and Multiagent Systems, 2019.

\bibitem[Castellini et~al.(2021)Castellini, Devlin, Oliehoek, and Savani]{ea}
Jacopo Castellini, Sam Devlin, Frans~A. Oliehoek, and Rahul Savani.
\newblock Difference rewards policy gradients.
\newblock In \emph{Proceedings of the 20th International Conference on
  Autonomous Agents and MultiAgent Systems}, AAMAS'21, pages 1475--1477.
  International Foundation for Autonomous Agents and Multiagent Systems, 2021.

\bibitem[Chang et~al.(2003)Chang, Ho, and Kaelbling]{local}
Yu-Han Chang, Tracey Ho, and Leslie~P. Kaelbling.
\newblock All learning is local: Multi-agent learning in global reward games.
\newblock In \emph{Advances in Neural Information Processing Systems 16},
  NIPS'03, pages 807--814. MIT Press, 2003.

\bibitem[Chung et~al.(2014)Chung, Gulcehre, Cho, and Bengio]{gru}
Junyoung Chung, Caglar Gulcehre, Kyung~Hyun Cho, and Yoshua Bengio.
\newblock Empirical evaluation of gated recurrent neural networks on sequence
  modeling.
\newblock In \emph{NIPS'14 Workshop on Deep Learning and Representation
  Learning}, NIPS'14. 2014.

\bibitem[Claus and Boutilier(1998)]{dynamics}
Caroline Claus and Craig Boutilier.
\newblock The dynamics of reinforcement learning in cooperative multiagent
  systems.
\newblock In \emph{Proceedings of the 15th/10th AAAI Conference on Artificial
  Intelligence/Innovative Applications of Artificial Intelligence},
  AAAI'98/IAAI'98, pages 746--752. American Association for Artificial
  Intelligence, 1998.

\bibitem[Colby et~al.(2014)Colby, Curran, Rebhuhn, and Tumer]{reward2}
Mitchell~K. Colby, William Curran, Carrie Rebhuhn, and Kagan Tumer.
\newblock Approximating difference evaluations with local knowledge.
\newblock In \emph{Proceedings of the 13th International Conference on
  Autonomous Agents and Multiagent Systems}, AAMAS'14, pages 1577--1578.
  International Foundation for Autonomous Agents and Multiagent Systems, 2014.

\bibitem[Colby et~al.(2015)Colby, Curran, and Tumer]{reward1}
Mitchell~K. Colby, William Curran, and Kagan Tumer.
\newblock Approximating difference evaluations with local information.
\newblock In \emph{Proceedings of the 14th International Conference on
  Autonomous Agents and Multiagent Systems}, AAMAS'15, pages 1659--1660.
  International Foundation for Autonomous Agents and Multiagent Systems, 2015.

\bibitem[Devlin and Kudenko(2011)]{theoretical}
Sam Devlin and Daniel Kudenko.
\newblock Theoretical considerations of potential-based reward shaping for
  multi-agent systems.
\newblock In \emph{{AAMAS}}, pages 225--232. International Foundation for
  Autonomous Agents and Multiagent Systems, 2011.

\bibitem[Devlin and Kudenko(2012)]{shaping}
Sam Devlin and Daniel Kudenko.
\newblock Dynamic potential-based reward shaping.
\newblock In \emph{Proceedings of the 11th International Conference on
  Autonomous Agents and Multiagent Systems}, AAMAS'12, pages 433--440.
  International Foundation for Autonomous Agents and Multiagent Systems, 2012.

\bibitem[Devlin et~al.(2014)Devlin, Yliniemi, Kudenko, and Tumer]{potential}
Sam Devlin, Logan Yliniemi, Daniel Kudenko, and Kagan Tumer.
\newblock Potential-based difference rewards for multiagent reinforcement
  learning.
\newblock In \emph{Proceedings of the 13th International Conference on
  Autonomous Agents and Multiagent Systems}, AAMAS'14, pages 165--172.
  International Foundation for Autonomous Agents and Multiagent Systems, 2014.

\bibitem[Eck et~al.(2015)Eck, Soh, Devlin, and Kudenko]{pomdp}
Adam Eck, Leen-Kiat Soh, Sam Devlin, and Daniel Kudenko.
\newblock Potential-based reward shaping for finite horizon online pomdp
  planning.
\newblock \emph{Autonomous Agents and Multi-Agent Systems}, 30:\penalty0
  403--445, 2015.

\bibitem[Foerster et~al.(2018)Foerster, Farquhar, Afouras, Nardelli, and
  Whiteson]{coma}
Jakob~N. Foerster, Gregory Farquhar, Triantafyllos Afouras, Nantas Nardelli,
  and Shimon Whiteson.
\newblock Counterfactual multi-agent policy gradients.
\newblock In \emph{Proceedings of the 32th AAAI Conference on Artificial
  Intelligence}, AAAI'18, pages 2974--2982. AAAI Press, 2018.

\bibitem[Fujimoto et~al.(2018)Fujimoto, van Hoof, and Meger]{lr}
Scott Fujimoto, Herke van Hoof, and David Meger.
\newblock Addressing function approximation error in actor-critic methods.
\newblock In \emph{Proceedings of the 36th International Conference on Machine
  Learning}, ICML'18, pages 1587--1596. PMLR, 2018.

\bibitem[Greensmith et~al.(2004)Greensmith, Bartlett, and Baxter]{variance}
Evan Greensmith, Peter~L. Bartlett, and Jonathan Baxter.
\newblock Variance reduction techniques for gradient estimates in reinforcement
  learning.
\newblock \emph{Journal of Machine Learning Research}, 5:\penalty0
  1471--–1530, 2004.

\bibitem[Guestrin et~al.(2002)Guestrin, Lagoudakis, and Parr]{coordinated}
Carlos Guestrin, Michail~G. Lagoudakis, and Ronald Parr.
\newblock Coordinated reinforcement learning.
\newblock In \emph{Proceedings of the 19th International Conference on Machine
  Learning}, ICML'02, pages 227--234. Morgan Kaufmann Publishers Inc., 2002.

\bibitem[Gupta et~al.(2017)Gupta, Egorov, and Kochenderfer]{cooperative}
Jayesh~K. Gupta, Maxim Egorov, and Mykel~J. Kochenderfer.
\newblock Cooperative multi-agent control using deep reinforcement learning.
\newblock \emph{Autonomous Agents and Multi-Agent Systems}, pages 66--83, 2017.

\bibitem[Hansen et~al.(2004)Hansen, Bernstein, and Zilberstein]{posg}
Eric~A. Hansen, Daniel~S. Bernstein, and Shlomo Zilberstein.
\newblock Dynamic programming for partially observable stochastic games.
\newblock In \emph{Proceedings of the 19th AAAI Conference on Artifical
  Intelligence}, AAAI'04, pages 709--715. AAAI Press, 2004.

\bibitem[Hernandez-Leal et~al.(2019)Hernandez-Leal, Kartal, and
  Taylor]{critique}
Pablo Hernandez-Leal, Bilal Kartal, and Matthew~E. Taylor.
\newblock A survey and critique of multiagent deep reinforcement learning.
\newblock \emph{Autonomous Agents and Multi-Agent Systems}, 33:\penalty0
  750--797, 2019.

\bibitem[Jaderberg et~al.(2016)Jaderberg, Mnih, Czarnecki, Schaul, Leibo,
  Silver, and Kavukcuoglu]{auxiliary}
Max Jaderberg, Volodymyr Mnih, Wojciech~M. Czarnecki, Tom Schaul, Joel~Z.
  Leibo, David Silver, and Koray Kavukcuoglu.
\newblock Reinforcement learning with unsupervised auxiliary tasks.
\newblock \emph{arXiv}, abs/1611.05397, 2016.

\bibitem[Kaelbling et~al.(1996)Kaelbling, Littman, and Moore]{survey}
Leslie~P. Kaelbling, Michael~L. Littman, and Andrew~W. Moore.
\newblock Reinforcement learning: A survey.
\newblock \emph{Journal of Artificial Intelligence Research}, 4\penalty0
  (1):\penalty0 237--285, 1996.

\bibitem[Konda and Tsitsiklis(2003)]{ac}
Vijay~R. Konda and John~N. Tsitsiklis.
\newblock On actor-critic algorithms.
\newblock \emph{SIAM Journal of Control and Optimization}, 42\penalty0
  (4):\penalty0 1143--1166, 2003.

\bibitem[Kraemer and Banerjee(2016)]{planning}
Landon Kraemer and Bikramjit Banerjee.
\newblock Multi-agent reinforcement learning as a rehearsal for decentralized
  planning.
\newblock \emph{Neurocomputing}, 190:\penalty0 82--94, 2016.

\bibitem[Lowe et~al.(2017)Lowe, Wu, Tamar, Harb, Abbeel, and Mordatch]{maddpg}
Ryan Lowe, Yi~Wu, Aviv Tamar, Jean Harb, Pieter Abbeel, and Igor Mordatch.
\newblock Multi-agent actor-critic for mixed cooperative-competitive
  environments.
\newblock In \emph{Advances in Neural Information Processing Systems 30},
  NIPS'17, pages 6379--6390. Curran Associates, Inc., 2017.

\bibitem[Matignon et~al.(2012)Matignon, Laurent, and Le~Fort-Piat]{independent}
Laetitia Matignon, Guillaume~J. Laurent, and Nadine Le~Fort-Piat.
\newblock Independent reinforcement learners in cooperative markov games: a
  survey regarding coordination problems.
\newblock \emph{Knowledge Engineering Review}, 27\penalty0 (1):\penalty0 1--31,
  2012.

\bibitem[Mnih et~al.(2015)Mnih, Kavukcuoglu, Silver, Rusu, Veness, Bellemare,
  Graves, Riedmiller, Fidjeland, Ostrovski, Petersen, Beattie, Sadik,
  Antonoglou, King, Kumaran, Wierstra, Legg, and Hassabis]{dqn}
Volodymyr Mnih, Koray Kavukcuoglu, David Silver, Andrei~A. Rusu, Joel Veness,
  Marc~G. Bellemare, Alex Graves, Martin Riedmiller, Andreas~K. Fidjeland,
  Georg Ostrovski, Stig Petersen, Charles Beattie, Amir Sadik, Ioannis
  Antonoglou, Helen King, Dharshan Kumaran, Daan Wierstra, Shane Legg, and
  Demis Hassabis.
\newblock Human-level control through deep reinforcement learning.
\newblock \emph{Nature}, 518\penalty0 (7540):\penalty0 529--533, 2015.

\bibitem[Mnih et~al.(2016)Mnih, Badia, Mirza, Graves, Lillicrap, Harley,
  Silver, and Kavukcuoglu]{a3c}
Volodymyr Mnih, Adrià~P. Badia, Mehdi Mirza, Alex Graves, Timothy Lillicrap,
  Tim Harley, David Silver, and Koray Kavukcuoglu.
\newblock Asynchronous methods for deep reinforcement learning.
\newblock In \emph{Proceedings 33rd International Conference on Machine
  Learning}, ICML'16, pages 1928--1937. PMLR, 2016.

\bibitem[Nair and Hinton(2010)]{relu}
Vinod Nair and Geoffrey~E. Hinton.
\newblock Rectified linear units improve restricted boltzmann machines.
\newblock In \emph{Proceedings of the 27th International Conference on
  International Conference on Machine Learning}, ICML'10, pages 807--814.
  Omnipress, 2010.

\bibitem[Ng et~al.(1999)Ng, Harada, and Russell]{invariance}
Andrew~Y. Ng, Daishi Harada, and Stuart Russell.
\newblock Policy invariance under reward transformations: Theory and
  application to reward shaping.
\newblock In \emph{Proceedings of the 16th International Conference on Machine
  Learning}, ICML'99, pages 278--287. Morgan Kaufmann, 1999.

\bibitem[Nguyen et~al.(2018)Nguyen, Kumar, and Lau]{credit}
Duc~T. Nguyen, Akshat Kumar, and Hoong~C. Lau.
\newblock Credit assignment for collective multiagent rl with global rewards.
\newblock In \emph{Advances in Neural Information Processing Systems 32},
  NIPS'18, pages 8113--8124. Curran Associates, Inc., 2018.

\bibitem[Nissim and Brafman(2012)]{maastar}
Raz Nissim and Ronen~I. Brafman.
\newblock Multi-agent a* for parallel and distributed systems.
\newblock In \emph{Proceedings of the 11th International Conference on
  Autonomous Agents and Multiagent Systems}, AAMAS'12, pages 1265--1266.
  International Foundation for Autonomous Agents and Multiagent Systems, 2012.

\bibitem[Oliehoek and Amato(2016)]{decpomdp}
Frans~A. Oliehoek and Christoper Amato.
\newblock \emph{A Concise Introduction to Decentralized POMDPs}.
\newblock Springer Publishing Company, Incorporated, 1st edition, 2016.

\bibitem[Papoudakis et~al.(2019)Papoudakis, Christianos, Rahman, and
  Albrecht]{ctde}
Georgios Papoudakis, Filippos Christianos, Arrasy Rahman, and Stefano~V.
  Albrecht.
\newblock Dealing with non-stationarity in multi-agent deep reinforcement
  learning.
\newblock \emph{arXiv}, abs/1906.04737, 2019.

\bibitem[Peshkin et~al.(2000)Peshkin, Kim, Meuleau, and Kaelbling]{mapg}
Leonid Peshkin, Kee-Eung Kim, Nicolas Meuleau, and Leslie~Pack Kaelbling.
\newblock Learning to cooperate via policy search.
\newblock In \emph{Proceedings of the 16th Conference on Uncertainty in
  Artificial Intelligence}, UAI’00, pages 489--–496. Morgan Kaufmann
  Publishers Inc., 2000.

\bibitem[Proper and Tumer(2012)]{difference}
Scott Proper and Kagan Tumer.
\newblock Modeling difference rewards for multiagent learning.
\newblock In \emph{Proceedings of the 11th International Conference on
  Autonomous Agents and Multiagent Systems}, AAMAS'12, pages 1397--1398.
  International Foundation for Autonomous Agents and Multiagent Systems, 2012.

\bibitem[Romoff et~al.(2018)Romoff, Henderson, Piche, Francois-Lavet, and
  Pineau]{reward}
Joshua Romoff, Peter Henderson, Alexandre Piche, Vincent Francois-Lavet, and
  Joelle Pineau.
\newblock Reward estimation for variance reduction in deep reinforcement
  learning.
\newblock In \emph{Proceedings of the 6th International Conference on Learning
  Representations}, ICLR'18, 2018.

\bibitem[Samvelyan et~al.(2019)Samvelyan, Rashid, Schr{\"{o}}eder~de Witt,
  Farquhar, Nardelli, Rudner, Hung, Torr, Foerster, and Whiteson]{smac}
Mikayel Samvelyan, Tabish Rashid, Christian Schr{\"{o}}eder~de Witt, Gregory
  Farquhar, Nantas Nardelli, Tim G.~J. Rudner, Chia-Man Hung, Philiph H.~S.
  Torr, Jakob Foerster, and Shimon Whiteson.
\newblock The starcraft multi-agent challenge.
\newblock \emph{arXiv}, abs/1902.04043, 2019.

\bibitem[Srinivasan et~al.(2018)Srinivasan, Lanctot, Zambaldi, P{\'e}rolat,
  Tuyls, Munos, and Bowling]{po}
Sriram Srinivasan, Marc Lanctot, Vinicius Zambaldi, Julien P{\'e}rolat, Karl
  Tuyls, Remi Munos, and Michael Bowling.
\newblock Actor-critic policy optimization in partially observable multiagent
  environments.
\newblock In \emph{Advances in Neural Information Processing Systems 32},
  NIPS'18, pages 3426--3439. Curran Associates Inc., 2018.

\bibitem[Sutton(1988)]{td}
Richard~S. Sutton.
\newblock Learning to predict by the methods of temporal differences.
\newblock \emph{Machine Learning}, 3\penalty0 (1):\penalty0 9--44, 1988.

\bibitem[Sutton and Barto(1998)]{rl}
Richard~S. Sutton and Andrew~G. Barto.
\newblock \emph{Introduction to Reinforcement Learning}.
\newblock MIT Press, 1st edition, 1998.

\bibitem[Sutton et~al.(2000)Sutton, McAllester, Singh, and Mansour]{pg}
Richard~S. Sutton, David~A. McAllester, Satinder~P. Singh, and Yishay Mansour.
\newblock Policy gradient methods for reinforcement learning with function
  approximation.
\newblock In \emph{Advances in Neural Information Processing Systems 12},
  NIPS'00, pages 1057--1063. MIT Press, 2000.

\bibitem[Tan(1993)]{qlearn}
Ming Tan.
\newblock Multi-agent reinforcement learning: Independent vs. cooperative
  agents.
\newblock In \emph{Proceedings of the 10th International Conference on Machine
  Learning}, ICML'93, pages 330--337. Morgan Kaufmann Publishers Inc., 1993.

\bibitem[Tumer and Agogino(2007)]{airflow}
Kagan Tumer and Adrian Agogino.
\newblock Distributed agent-based air traffic flow management.
\newblock In \emph{Proceedings of the 6th International Conference on
  Autonomous Agents and Multiagent Systems}, AAMAS'07. Association for
  Computing Machinery, 2007.

\bibitem[Van~der Pol and Oliehoek(2016)]{traffic}
Elise Van~der Pol and Frans~A. Oliehoek.
\newblock Coordinated deep reinforcement learners for traffic light control.
\newblock In \emph{NIPS'16 Workshop on Learning, Inference and Control of
  Multi-Agent Systems}, NIPS'16. 2016.

\bibitem[Vinyals et~al.(2017)Vinyals, Ewalds, Bartunov, Georgiev, Vezhnevets,
  Yeo, Makhzani, K{\"{u}}ttler, Agapiou, Schrittwieser, Quan, Gaffney,
  Petersen, Simonyan, Schaul, van Hasselt, Silver, Lillicrap, Calderone, Keet,
  Brunasso, Lawrence, Ekermo, Repp, and Tsing]{sc2}
Oriol Vinyals, Timo Ewalds, Sergey Bartunov, Petko Georgiev, Alexander~(Sasha)
  Vezhnevets, Michelle Yeo, Alireza Makhzani, Heinrich K{\"{u}}ttler, John~P.
  Agapiou, Julian Schrittwieser, John Quan, Stephen Gaffney, Stig Petersen,
  Karen Simonyan, Tom Schaul, Hado van Hasselt, David Silver, Timothy~P.
  Lillicrap, Kevin Calderone, Paul Keet, Anthony Brunasso, David Lawrence,
  Anders Ekermo, Jacob Repp, and Rodney Tsing.
\newblock {StarCraft II}: A new challenge for reinforcement learning.
\newblock \emph{arXiv}, abs/1708.04782, 2017.

\bibitem[Wang et~al.(2020)Wang, Han, Wang, Dong, and Zhang]{dop}
Yihan Wang, Beining Han, Tonghan Wang, Heng Dong, and Chongjie Zhang.
\newblock Off-policy multi-agent decomposed policy gradients.
\newblock \emph{arXiv}, abs/2007.12322, 2020.

\bibitem[Williams(1992)]{reinforce}
Ronald~J. Williams.
\newblock Simple statistical gradient-gollowing algorithms for connectionist
  reinforcement learning.
\newblock \emph{Machine Learning}, 8\penalty0 (3), 1992.

\bibitem[Wolpert and Tumer(1999)]{coin}
David~H. Wolpert and Kagan Tumer.
\newblock An introduction to collective intelligence.
\newblock Technical report, NASA-ARC-IC-99-63, Nasa Ames Research Center, 1999.

\bibitem[Wolpert and Tumer(2001)]{aristocrat}
David~H. Wolpert and Kagan Tumer.
\newblock Optimal payoff functions for members of collectives.
\newblock \emph{Advances in Complex Systems}, 4:\penalty0 265--280, 2001.

\bibitem[Ye et~al.(2015)Ye, Zhang, and Yang]{sensors}
Dayon Ye, Minji Zhang, and Yu~Yang.
\newblock A multi-agent framework for packet routing in wireless sensor
  networks.
\newblock \emph{Sensors}, 15\penalty0 (5):\penalty0 10026--10047, 2015.

\bibitem[Yliniemi and Tumer(2014)]{objective}
Logan Yliniemi and Kagan Tumer.
\newblock Multi-objective multiagent credit assignment through difference
  rewards in reinforcement learning.
\newblock In \emph{Asia-Pacific Conference on Simulated Evolution and
  Learning}, pages 407--418. Springer International Publishing, 2014.

\bibitem[Zhang and Zavlanos(2019)]{consensus}
Yan Zhang and Michael~M. Zavlanos.
\newblock Distributed off-policy actor-critic reinforcement learning with
  policy consensus.
\newblock \emph{arXiv}, abs/1903.09255, 2019.

\end{thebibliography}

\clearpage
\begin{appendices}
\section{Hyperparameters and Training}
For our implementation, we relied on and expanded the \texttt{pymarl} \citep{smac} framework, as already providing many useful tools and the official implementation of COMA to compare against. The policy networks are either feedforward networks for the two gridworld problems or GRU \citep{gru} to deal with partial observability on SMAC, and both use parameter sharing across agents \citep{cooperative} to reduce training time, while the critics and reward networks use feedforward networks instead.

On each of the three problems independently, the optimal values for policy learning rate $\alpha_{\theta}$ and the critic or reward network one $\alpha_{\omega/\psi}$ \citep{lr} have been found for each method through a gridsearch over a common set of standard values. We used the setting with $N=3$ agents for the two gridworld environments and the map \texttt{2s3z} on SMAC, and the values obtained this way have been subsequently used for the other instances of the same problem respectively. Table \ref{tab:lr} reports the value of the used learning rates $\alpha_{\theta}$ and $\alpha_{\omega/\psi}$ for each compared method on each problem.

\begin{table}[htbp]
\begin{minipage}{\textwidth}
\begin{center}
\caption{Value of the learning rates for each method.}
\label{tab:lr}
\begin{tabular}{lrrrrrr}
\toprule
 & \multicolumn{2}{c}{Multi-Rover} & \multicolumn{2}{c}{Predator-Prey} & \multicolumn{2}{c}{SMAC} \\
Method & $\alpha_{\theta}$ & $\alpha_{\omega/\psi}$ & $\alpha_{\theta}$ & $\alpha_{\omega/\psi}$ & $\alpha_{\theta}$ & $\alpha_{\omega/\psi}$ \\
\midrule
Dr.Reinforce & $25\cdot 10^{-4}$ & N.A. & $25\cdot 10^{-4}$ & N.A. & N.A. & N.A. \\
Dr.ReinforceR & $25\cdot 10^{-4}$ & $25\cdot 10^{-4}$ & $5\cdot 10^{-4}$ & $25\cdot 10^{-4}$ & $25\cdot 10^{-4}$ & $25\cdot 10^{-4}$ \\
COMA & $1\cdot 10^{-2}$ & $5\cdot 10^{-4}$ & $1\cdot 10^{-2}$ & $5\cdot 10^{-4}$ & $25\cdot 10^{-4}$ & $5\cdot 10^{-4}$ \\
\citep{reward2} & $5\cdot 10^{-3}$ & $25\cdot 10^{-4}$ & $5\cdot 10^{-4}$ & $1\cdot 10^{-2}$ & $5\cdot 10^{-4}$ & $5\cdot 10^{-4}$ \\
CentralQ & $5\cdot 10^{-4}$ & $25\cdot 10^{-4}$ & $5\cdot 10^{-4}$ & $5\cdot 10^{-3}$ & $5\cdot 10^{-4}$ & $5\cdot 10^{-4}$ \\
PG & $5\cdot 10^{-4}$ & N.A. & $5\cdot 10^{-4}$ & N.A. & $5\cdot 10^{-4}$ & N.A. \\
\bottomrule
\end{tabular}
\end{center}
\end{minipage}
\end{table}

COMA \citep{coma} and CentralQ critics have been trained using the TD($\lambda$) \citep{td,rl} variant presented in \citep{coma}. For these, the optimal value for the parameter $\lambda$ with the learning rates already found by the gridsearches has also been assessed following the same procedure detailed above, resulting in the values in Table \ref{tab:lambda}:

\begin{table}[htbp]
\begin{minipage}{\textwidth}
\begin{center}
\caption{Value of $\lambda$ for each method.}
\label{tab:lambda}
\begin{tabular}{lrrrrrr}
\toprule
Method & Multi-Rover & Predator-Prey & SMAC \\
\midrule
COMA & $0.4$ & $0.8$ & $0.8$ \\
CentralQ & $0.2$ & $0.8$ & $0.8$ \\
\bottomrule
\end{tabular}
\end{center}
\end{minipage}
\end{table}

All the methods have been trained for the same amount of steps and all their other hyperparameters are set to the corresponding default values provided by the \texttt{pymarl} framework, without being optimized: the reward network $R_{\psi}$ and the critic network $Q_{\omega}$ for CentralQ and COMA all have the same structure, that is a two-layer feedforward neural network with $128$ hidden units using the ReLU activation function \citep{relu} before the final linear layer, as the size of the functions these have to represent is analogous. Every experiment has been repeated $10$ times with different random seeds to assess variance across multiple runs, and in each episode the initial configuration has been randomly reset to avoid the policies to overfit.

\clearpage
\section{Statistical Significance Tests}
To assess the statistical significance of the proposed results, we computed a t-test on each algorithms' pair. The tested null hypothesis is that the samples (the return obtained by the different methods) are taken from the same distribution, meaning that any difference in the corresponding plotted lines are solely due to statistical noise rather than on the different capabilities of the algorithms. The test has been corrected with the Bonferroni correction term \citep{bonferroni} to account for the possible errors across the different pairings. Test results are reported in Figure \ref{fig:tests}, where a $\mathbf{+}$ symbol on a given cell means that the test value for a given algorithms' pair $p>0.05$. It is to note that results on the diagonal are obtained pairing a method with itself, and thus are clearly statistically correlated.

\begin{figure}[htbp]
\centering
\subfloat[Multi-Rover, $N=3$]{
\includegraphics[width=0.35\textwidth]{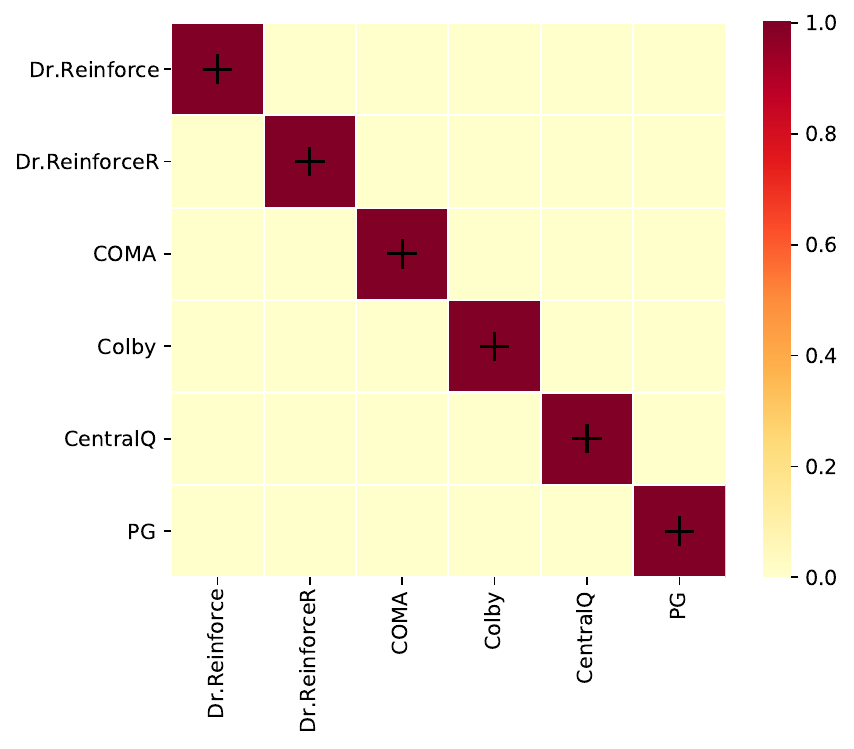}
}
\subfloat[Predator-Prey, $N=3$]{
\includegraphics[width=0.35\textwidth]{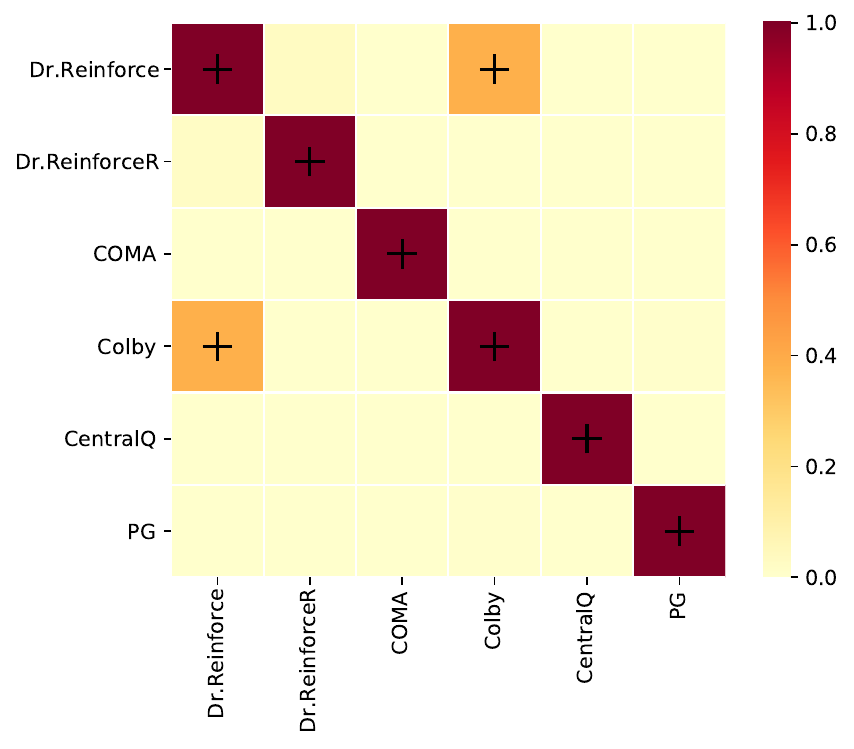}
}

\subfloat[Multi-Rover, $N=5$]{
\includegraphics[width=0.35\textwidth]{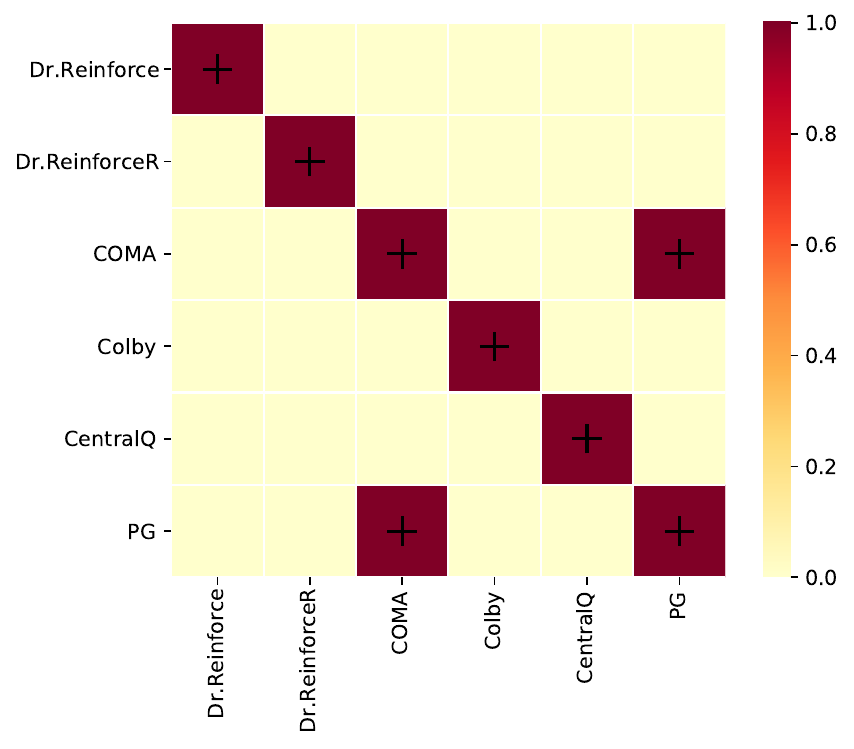}
}
\subfloat[Predator-Prey, $N=5$]{
\includegraphics[width=0.35\textwidth]{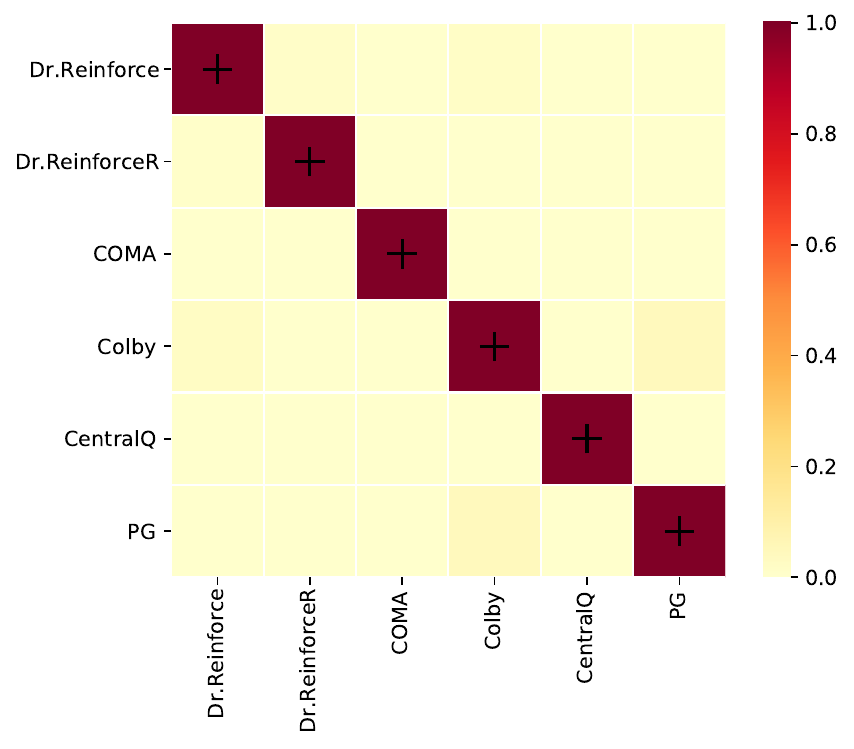}
}

\subfloat[Multi-Rover, $N=8$]{
\includegraphics[width=0.35\textwidth]{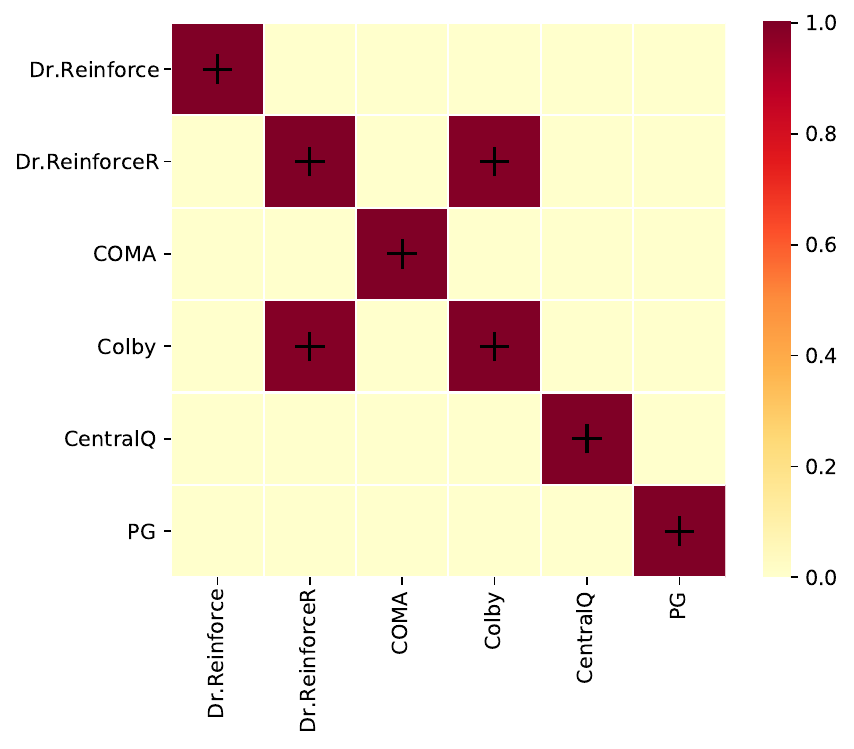}
}
\subfloat[Predator-Prey, $N=8$]{
\includegraphics[width=0.35\textwidth]{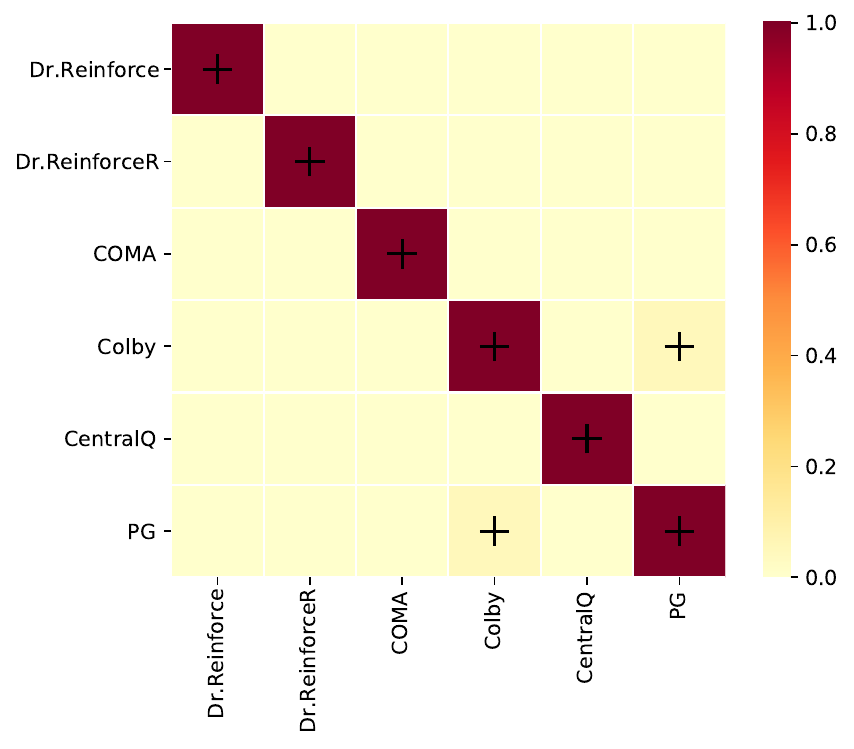}
}
\caption{Results of the t-test for different methods' pairs, corrected using the Bonferroni correction term,  on each problem instance.}
\label{fig:tests}
\end{figure}

\clearpage
\section{Additional Analysis Plots}
\begin{figure}[htbp]
\centering
\subfloat[Multi-Rover]{
\includegraphics[width=0.38\textwidth]{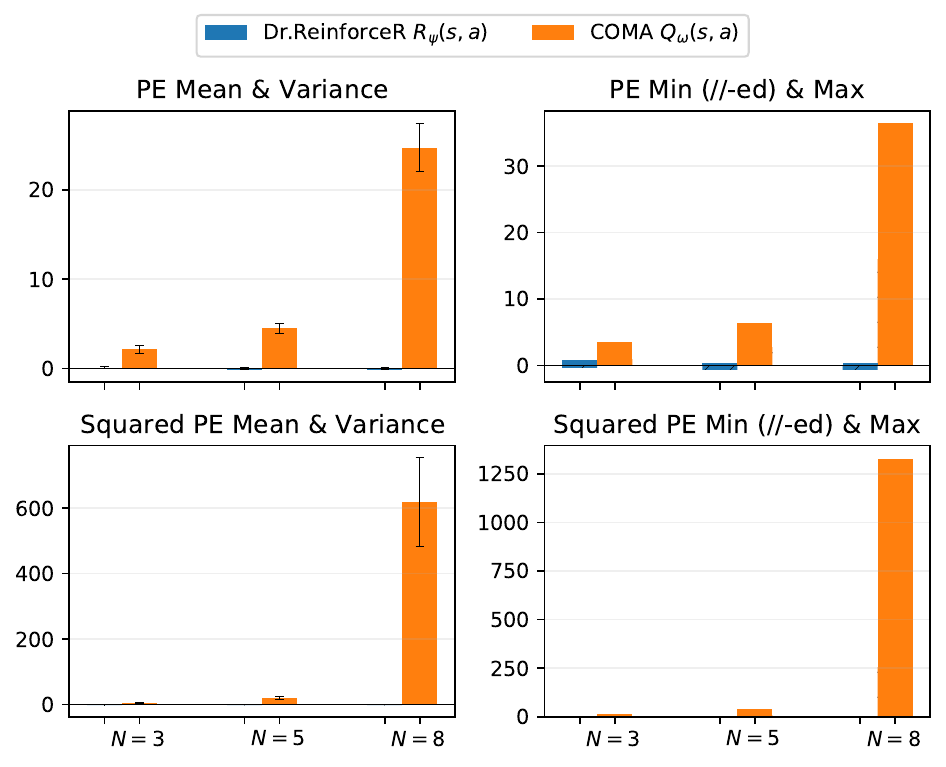}
}
\subfloat[Predator-Prey]{
\includegraphics[width=0.38\textwidth]{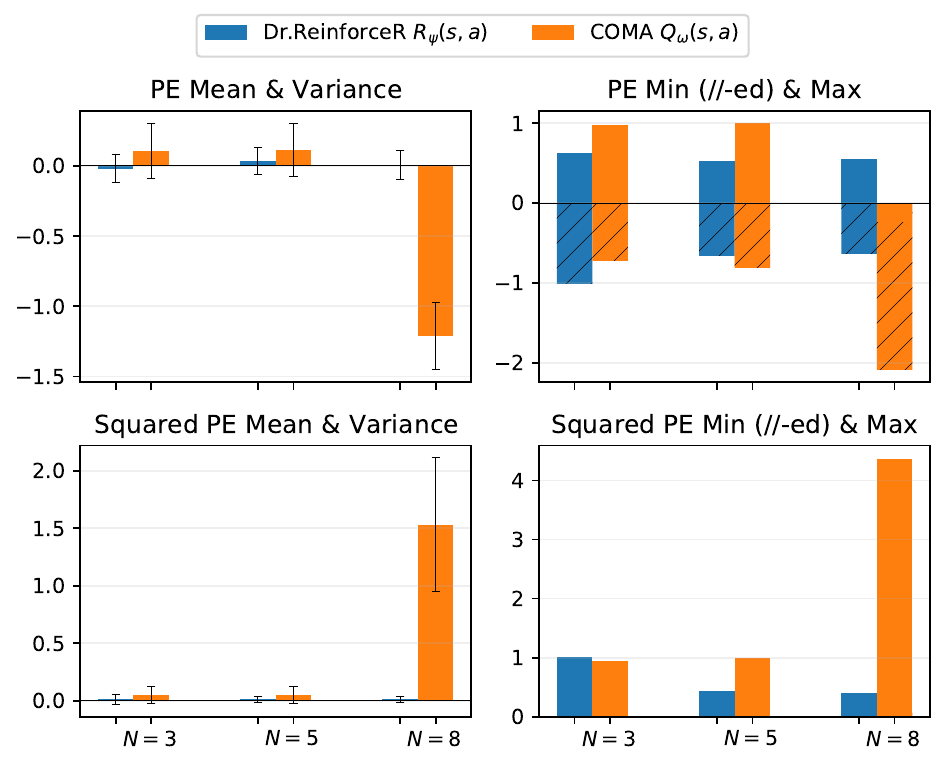}
}
\caption{Distribution statistics for Dr.ReinforceR reward network $R_{\psi}$ and COMA critic $Q_{\omega}$ on the on-policy dataset, normalized by the value of $r_{max}-r_{min}$ (respectively $q_{max}-q_{min}$ for COMA critic), for the two environments.}
\end{figure}

\begin{figure}[htbp]
\centering
\subfloat[Multi-Rover]{
\includegraphics[width=0.38\textwidth]{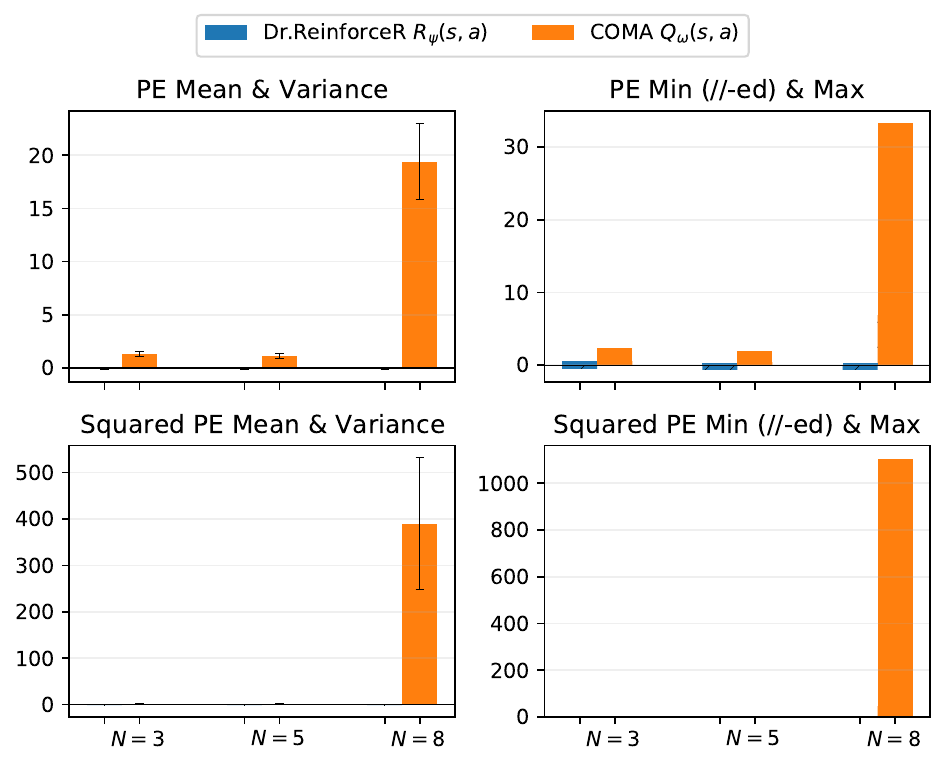}
}
\subfloat[Predator-Prey]{
\includegraphics[width=0.38\textwidth]{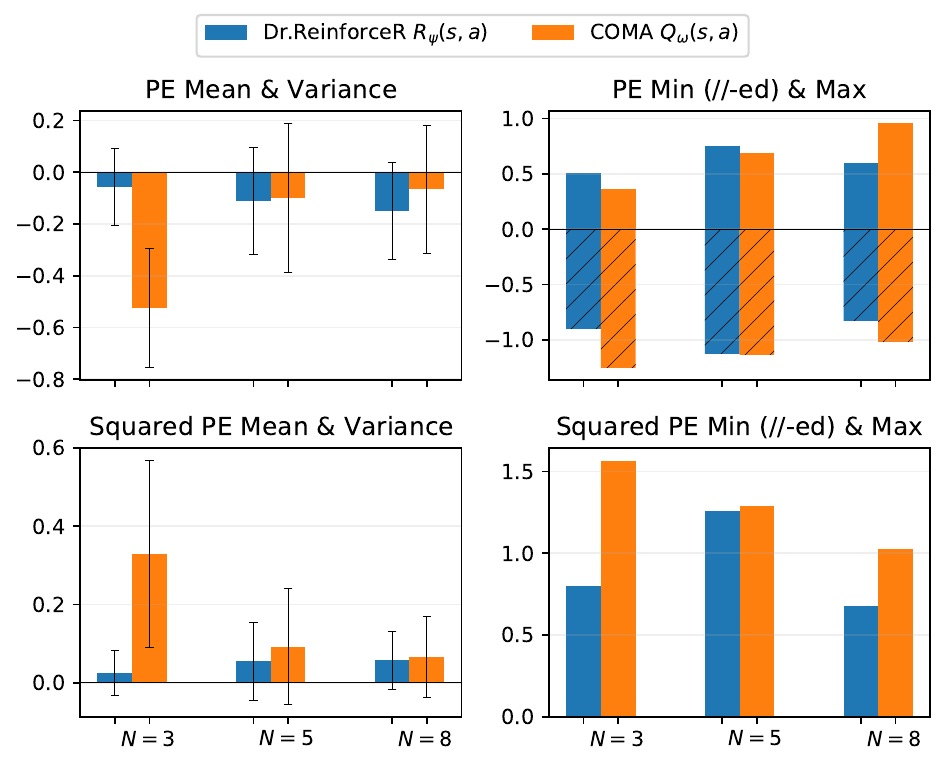}
}
\caption{Distribution statistics for Dr.ReinforceR reward network $R_{\psi}$ and COMA critic $Q_{\omega}$ on the off-policy dataset, normalized by the value of $r_{max}-r_{min}$ (respectively $q_{max}-q_{min}$ for COMA critic), for the two environments.}
\end{figure}

\begin{figure}[H]
\centering
\includegraphics[width=0.52\textwidth]{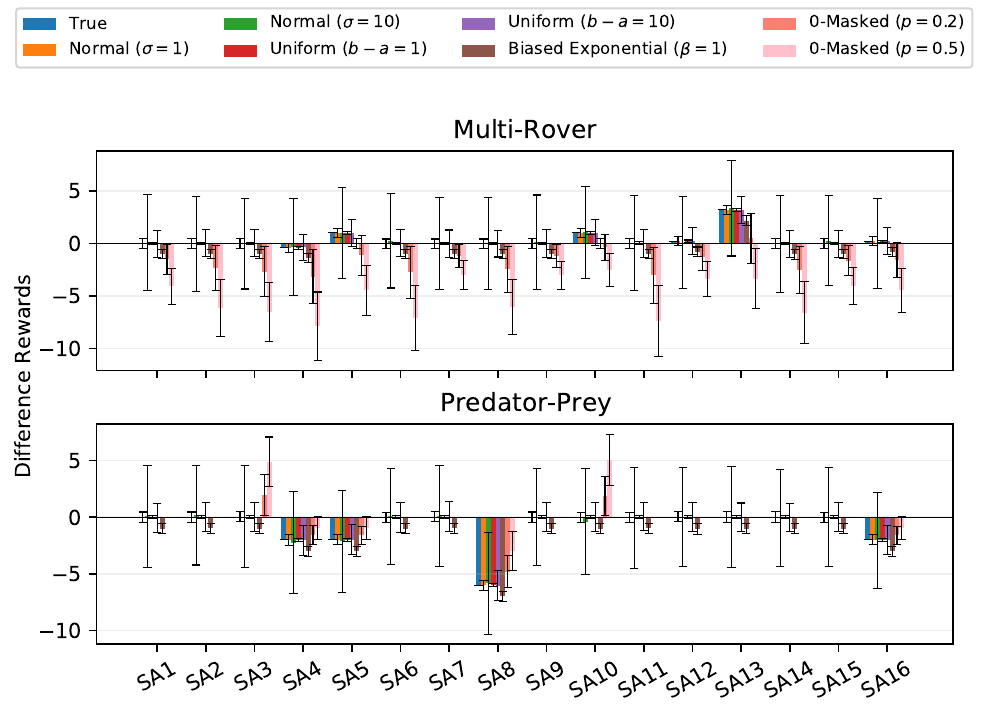}
\caption{Mean and variance of difference rewards for a set of samples under different noise profiles.}
\end{figure}

\clearpage
\section{Additional SMAC Plots}
\label{sec:smac}
\begin{figure}[H]
\vspace{-2\baselineskip}
\centering
\subfloat{
\includegraphics[width=0.4\textwidth]{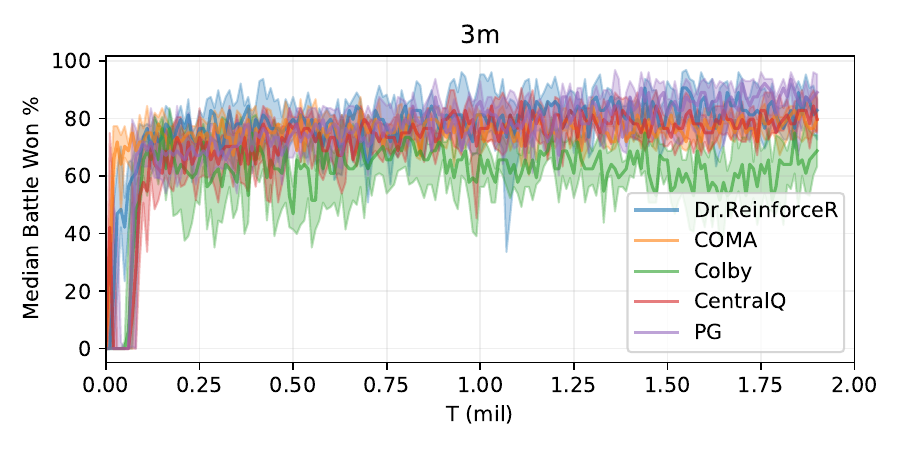}
}
\subfloat{
\includegraphics[width=0.4\textwidth]{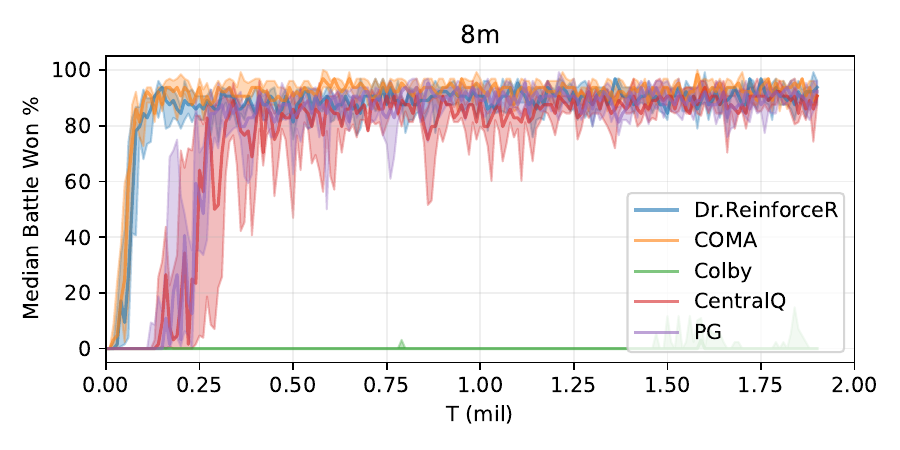}
}

\subfloat{
\includegraphics[width=0.4\textwidth]{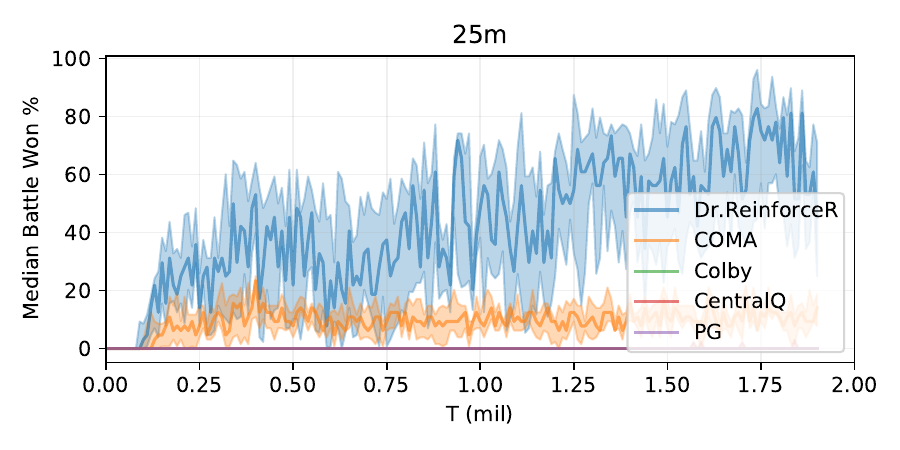}
}
\subfloat{
\includegraphics[width=0.4\textwidth]{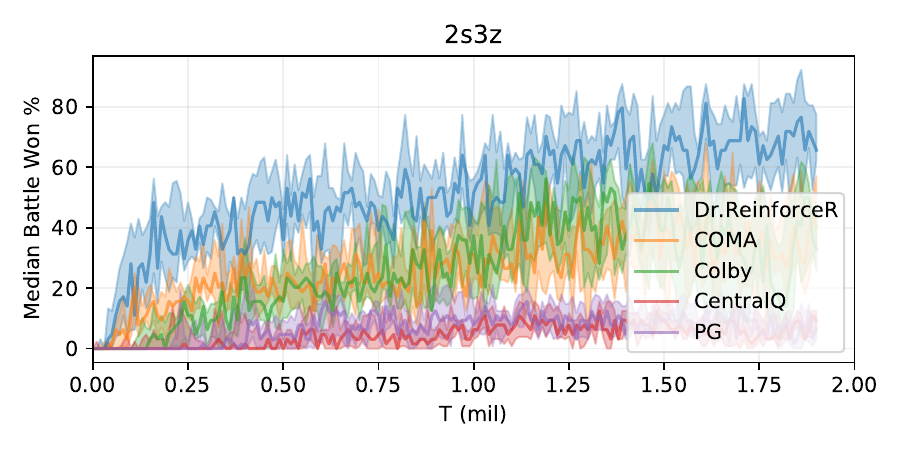}
}

\subfloat{
\includegraphics[width=0.4\textwidth]{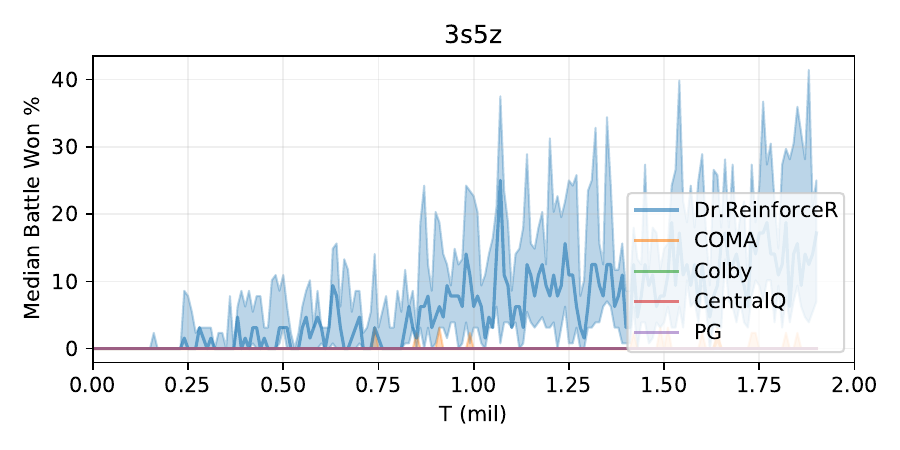}
}
\subfloat{
\includegraphics[width=0.4\textwidth]{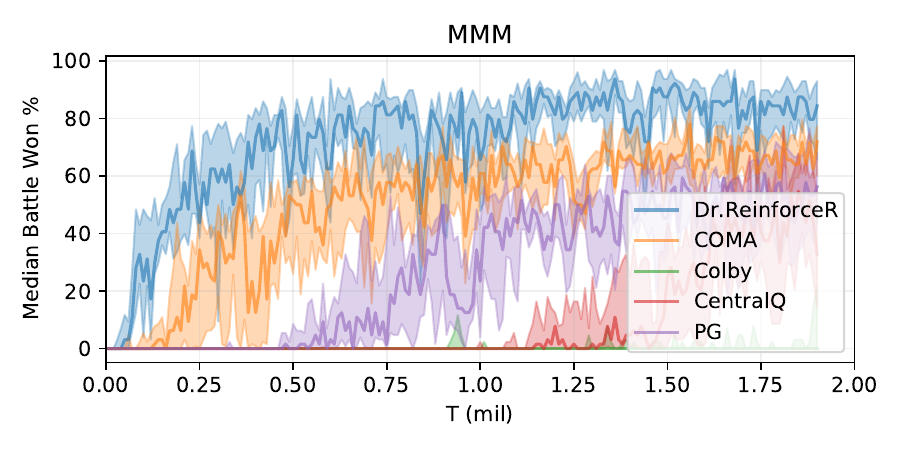}
}

\subfloat{
\includegraphics[width=0.4\textwidth]{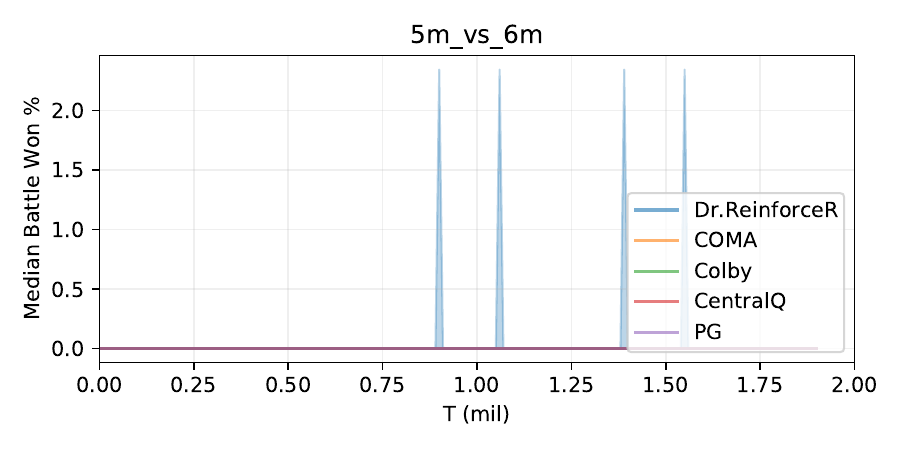}
}
\subfloat{
\includegraphics[width=0.4\textwidth]{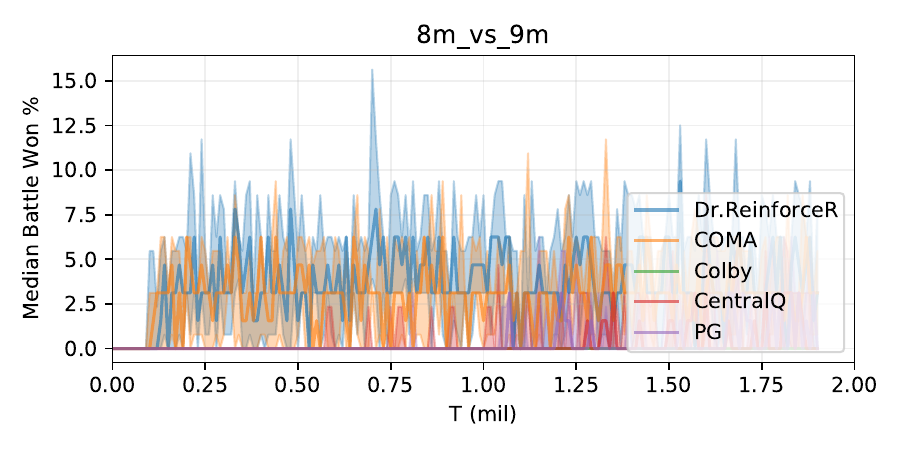}
}

\subfloat{
\includegraphics[width=0.4\textwidth]{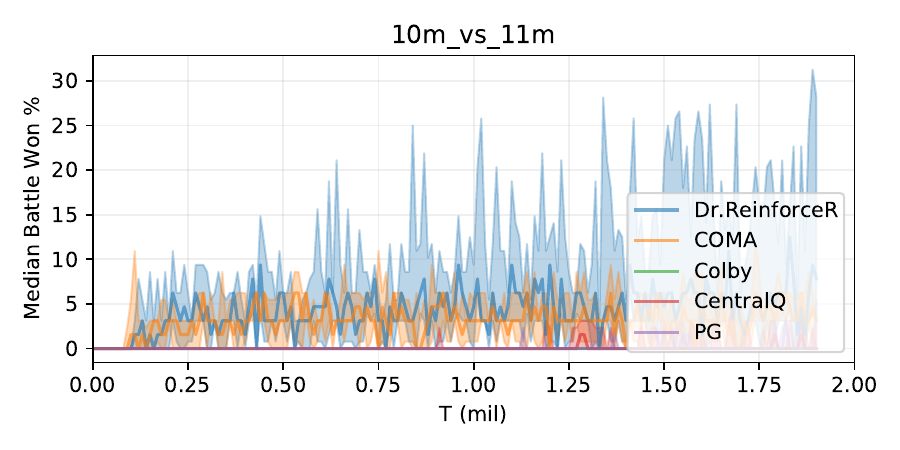}
}
\subfloat{
\includegraphics[width=0.4\textwidth]{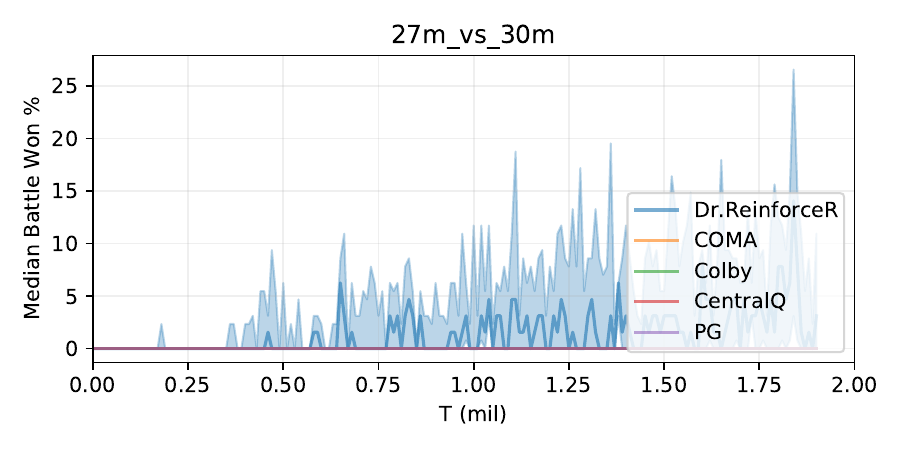}
}

\subfloat{
\includegraphics[width=0.4\textwidth]{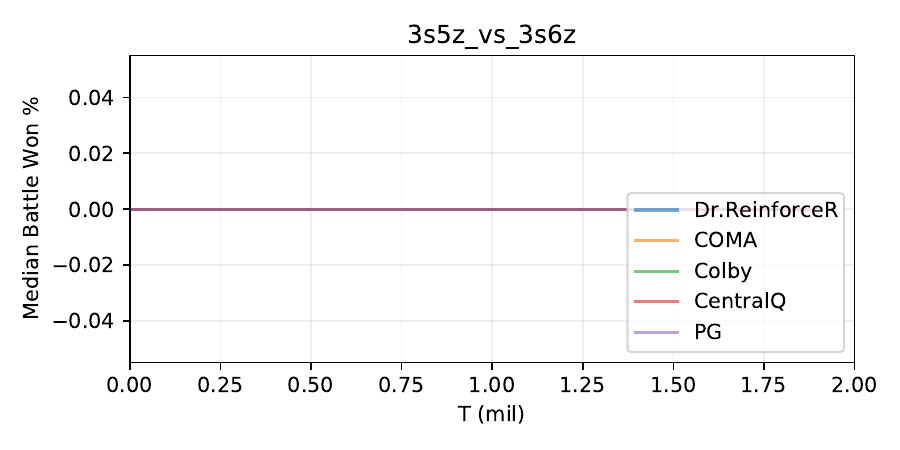}
}
\subfloat{
\includegraphics[width=0.4\textwidth]{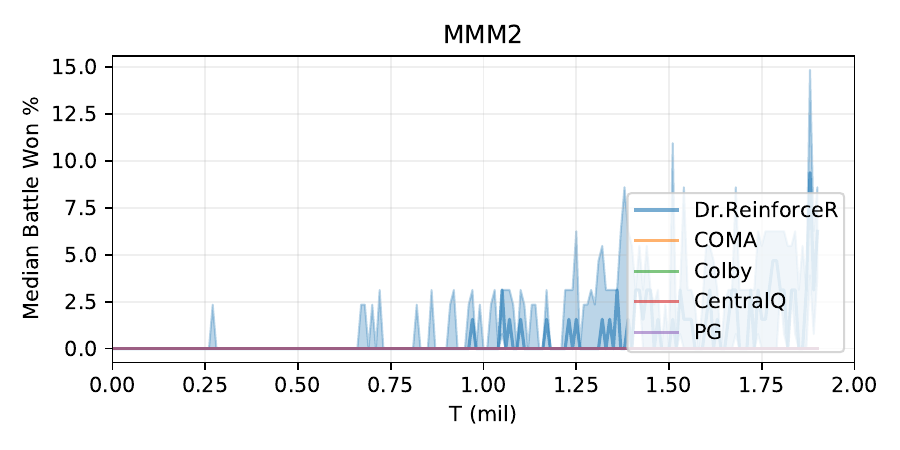}
}
\end{figure}

\begin{figure}[H]
\vspace{-2\baselineskip}
\centering
\subfloat{
\includegraphics[width=0.42\textwidth]{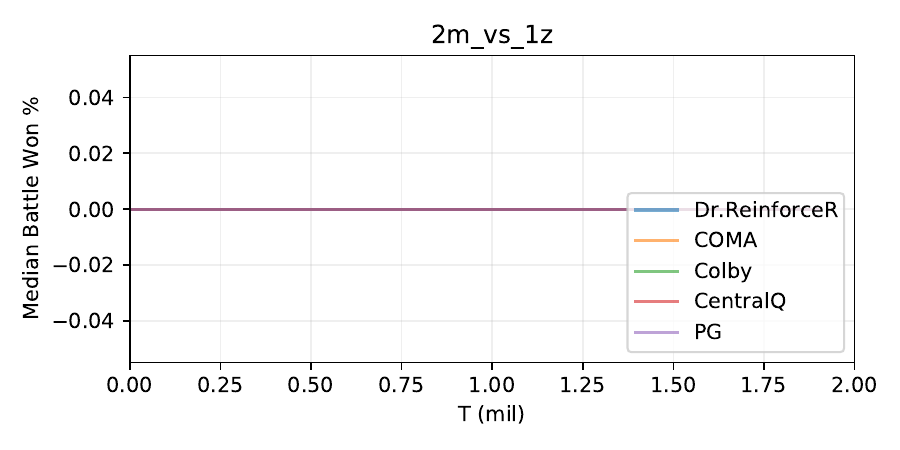}
}
\subfloat{
\includegraphics[width=0.42\textwidth]{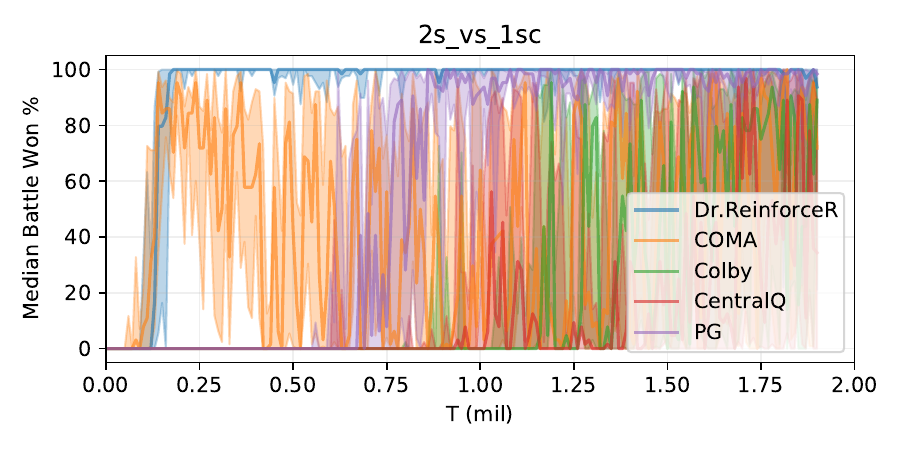}
}

\subfloat{
\includegraphics[width=0.42\textwidth]{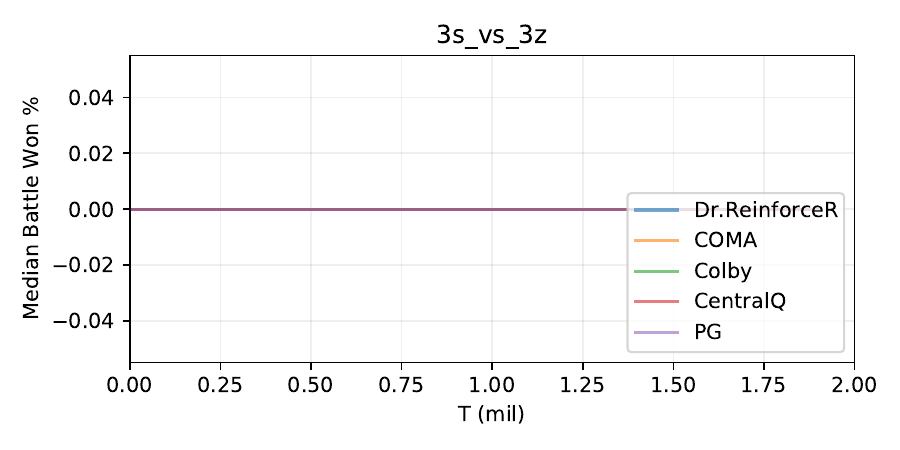}
}
\subfloat{
\includegraphics[width=0.42\textwidth]{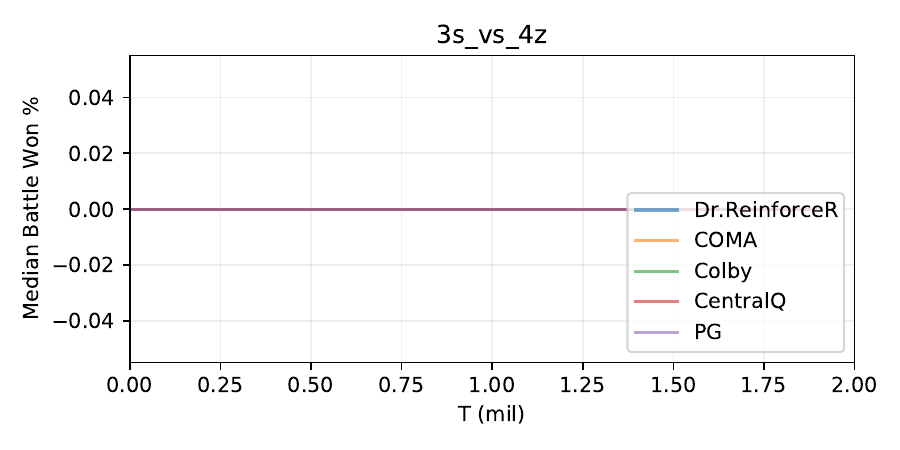}
}

\subfloat{
\includegraphics[width=0.42\textwidth]{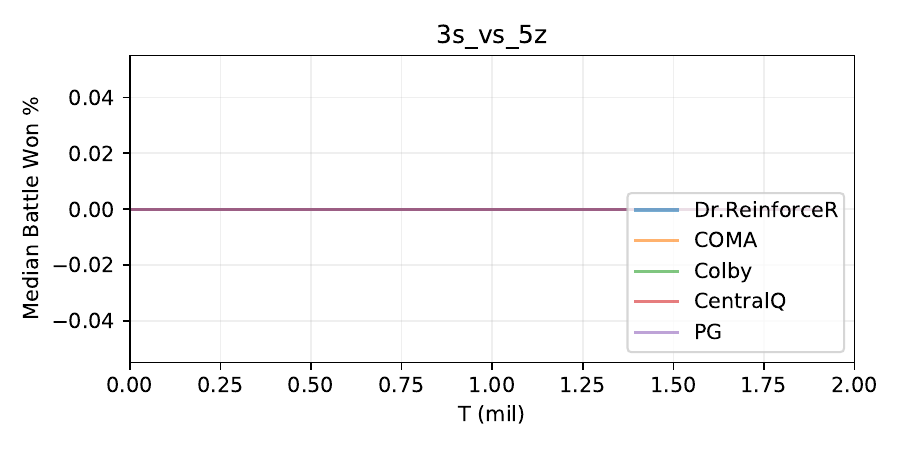}
}
\subfloat{
\includegraphics[width=0.42\textwidth]{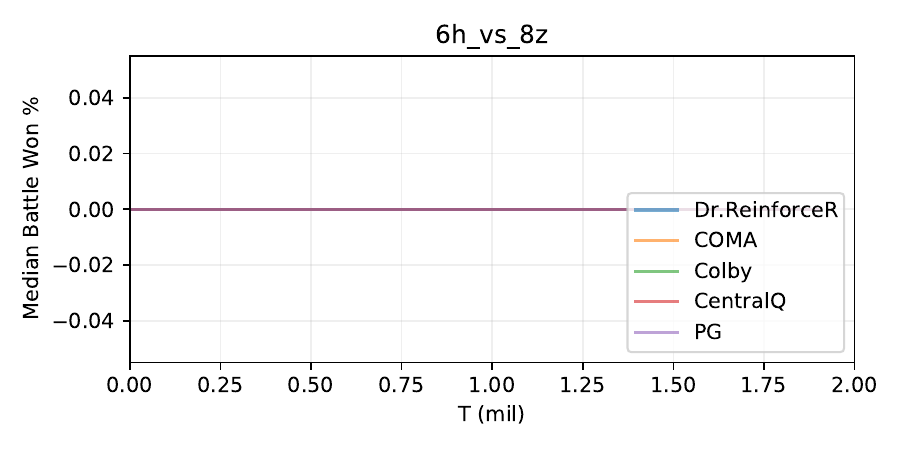}
}

\subfloat{
\includegraphics[width=0.42\textwidth]{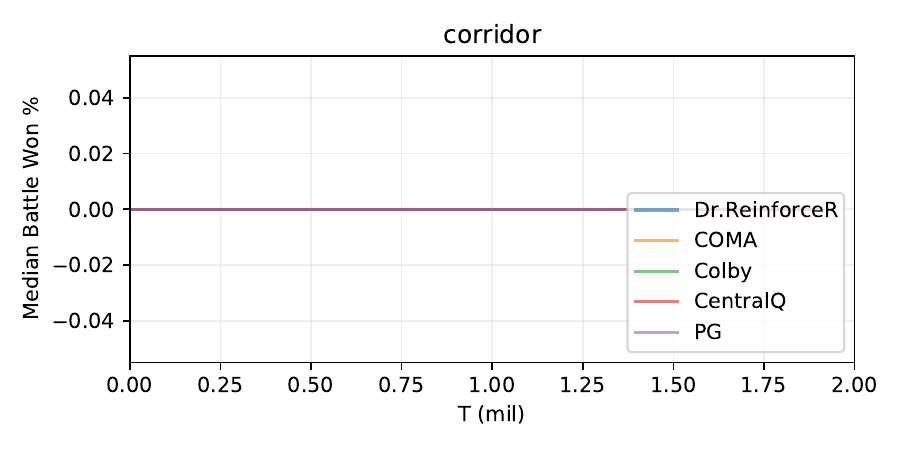}
}
\subfloat{
\includegraphics[width=0.42\textwidth]{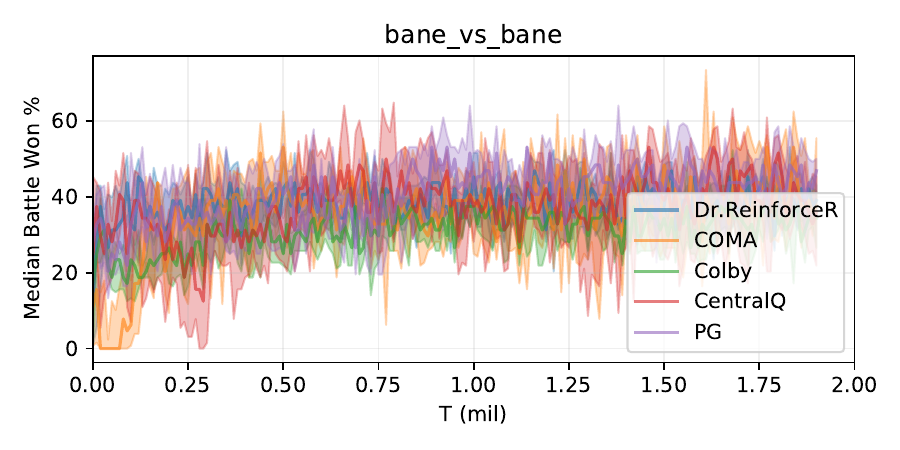}
}

\subfloat{
\includegraphics[width=0.42\textwidth]{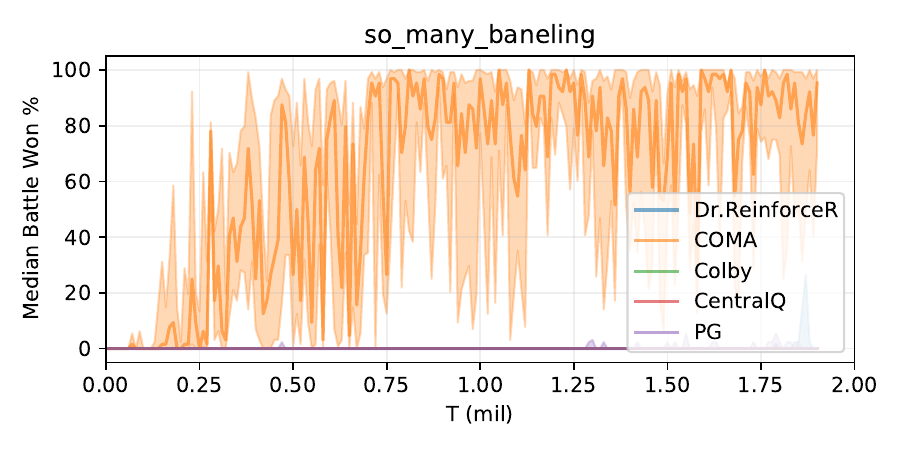}
}
\subfloat{
\includegraphics[width=0.42\textwidth]{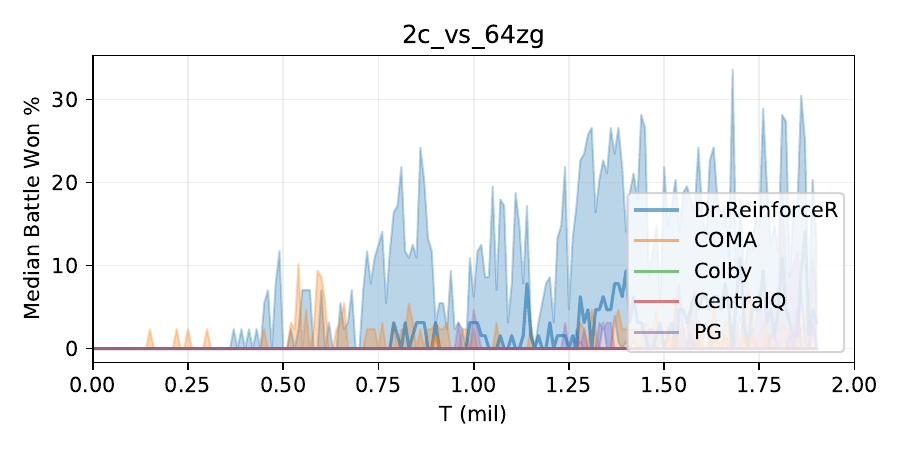}
}

\subfloat{
\includegraphics[width=0.42\textwidth]{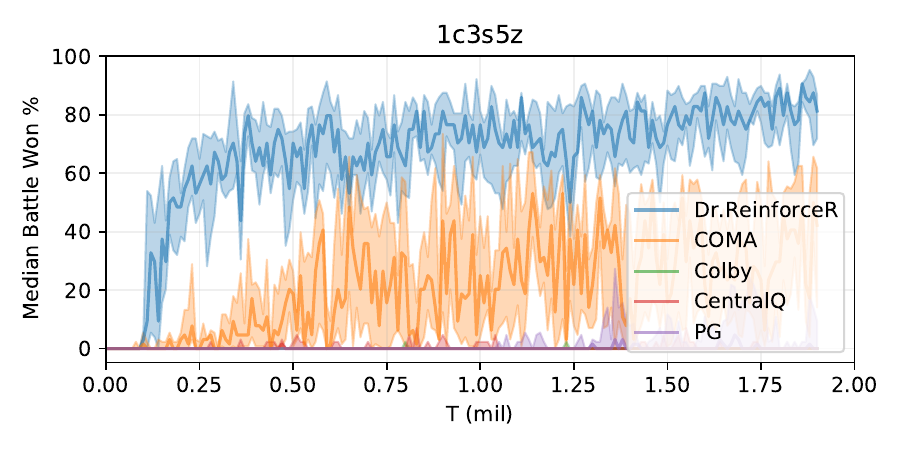}
}
\caption{Training curves on the entire set of SMAC maps, showing the median victory rate and $25-75\%$ percentiles across seeds.}
\end{figure}
\end{appendices}
\end{document}